\documentclass[a4paper,12pt]{article}
\usepackage{ulem}
\usepackage{cancel}
\usepackage{amsmath}
\usepackage{amsthm}
\usepackage{amsfonts}
\usepackage{mathrsfs}
\usepackage{graphicx}
\usepackage{hyperref}
\usepackage[dvipsnames]{xcolor}
\usepackage{enumerate}
\usepackage{tabularx}
\usepackage{tensor}
\usepackage{bm}
\usepackage{subcaption}
\usepackage{tikz}
\usepackage{float}

\newcommand{\lR}{\mathrm{I\hspace{-0.7mm}R}}

\newtheorem{theorem}{Theorem}
\newtheorem{lemma}{Lemma}

\numberwithin{equation}{section}

\hypersetup{ pdftitle={Draft},
pdfauthor={Fiki T. Akbar, Bobby E. Gunara}, 
pdfkeywords={Local Esistence, Dyonic Black Hole}, 
bookmarksnumbered, pdfstartview={FitH}, urlcolor=blue, }

\begin{document}
\pagestyle{plain}




\title{\LARGE\textbf{Some Cosmological Consequences of Higher Dimensional Klein-Gordon-Rastall  Theory}}

\author{\normalsize{Tegar Ari Widianto,$^{1}$} Ahmad Khoirul Falah,$^{1}$ Agus Suroso,$^{1}$\\ \normalsize{Husin Alatas,$^{2}$} and Bobby Eka Gunara$^{1}$\footnote{Corresponding author} \\
\textit{\small $^{1}$Theoretical Physics Laboratory,}
\textit{\small Theoretical High Energy Physics Group,}\\
\textit{\small Faculty of Mathematics and Natural Sciences,}
\textit{\small Institut Teknologi Bandung,}\\
\textit{\small Jl. Ganesha no. 10 Bandung, Indonesia, 40132}\\
\textit{\small $^{2}$Theoretical Physics Division, Department of Physics, IPB University,}\\
\textit{\small Jl. Meranti, Kampus IPB Darmaga, Bogor 16680, Indonesia}
\\ \\
\small email: tegarariwidianto76@gmail.com, akhoirulfalah94@students.itb.ac.id,\\
\small suroso@itb.ac.id,  alatas@apps.ipb.ac.id, bobby@itb.ac.id }

\date{\today}

\maketitle




\begin{abstract}
 Using dynamical system analysis, we investigate some cosmological consequences of Rastall gravity coupled to a scalar field (called the Klein-Gordon-Rastall theory) with exponential scalar potential turned on in higher dimensions.  From the critical points of the autonomous equations, we can determine the dominant components of the energy density in different cosmic eras.  We obtain a fixed point representing a scalar field-matter-dominated era which corresponds to either a late-time or past-time attractor depending on the parameters used. According to this point, the inflationary phase, corresponding to past-time attractors, is given by unstable nodes, whilst the dark energy era, corresponding to late-time attractors, is represented by stable nodes. In the inflationary sector, power-law inflation can still occur in this Klein-Gordon-Rastall cosmological model.  On the other hand, in the late-time sector, we find a nontrivial interplay between a scalar field with an exponential potential and the non-conservative energy-momentum tensor of the non-relativistic matter field (baryonic-dark matter) in curved spacetime plays a role as the dark energy. Based on such features, the Klein-Gordon-Rastall cosmology could be a promising candidate for describing both the early and late-time universe.

\end{abstract}




\section{Introduction}
\label{sec:intro}


The accuracy of General Relativity's predictions compared to the observational results has been confirmed at the level of the solar system \cite{will2006confrontation,bleeker2001century}, galaxies and galaxy clusters scale tests \cite{liu2022galaxy,wojtak2011gravitational}. However, some problems arise when we apply General Relativity (GR) to the cosmological scale to study the dynamics of the universe. To obtain an accelerating universe that is in agreement with the observational result, the cosmological constant, $\Lambda$, must be added to the Einstein field equation{\color{blue}s}, the $\Lambda$CDM model \cite{riess1998observational,lukash2000cosmological}. The problem lies in the interpretation of $\Lambda$ as a vacuum energy density \cite{weinberg1989cosmological, rugh2002quantum}. It leads to the cosmological constant problem where the calculation using quantum field theory shows that the vacuum energy density is $\rho_{vac}\approx 10^{74}$ GeV$^4$, a value that is much bigger than the cosmological observations, $\rho_{obs}\approx 10^{-47}$ GeV$^4$ \cite{weinberg1989cosmological, martin2012everything,copeland2006dynamics}. Thus, the interpretation of $\Lambda$ as a vacuum energy density remains questionable. It may be either a natural property of the spacetime fabric structure that has anti-gravity property or a mathematical representation such as Lagrange multiplier or the integration's constant \cite{nobbenhuis2006cosmological}. In other words, the source of cosmic acceleration is still unknown.

 There have been several attempts to overcome such a problem. In the late-time epoch, one of the proposed models is the dynamical scalar field with a particular potential form. The model was later referred to as \textit{quintessence} \cite{caldwell1998cosmological}, a canonical scalar field coupled to gravity that plays a role as dark energy which can explain the late-time cosmic acceleration. At early times, the scalar field varying slowly along the potential $V(\phi)$ can trigger the universe to accelerate expansion known as the inflationary era. Such a mechanism is called \textit{slow-roll inflation} where the scalar field plays the role of the inflaton field. The difference between those two eras is that the current accelerating universe, driven by dark energy, also contains dark matter and baryons, unlike the inflationary phase, which contains only the inflaton field. Other scalar field models such as k-essence, phantom, tachyon, and dilaton \cite{armendariz2000dynamical, hannestad2002probing, dabrowski2003phantom, gibbons2002cosmological, bagla2003cosmology, gasperini2008dilaton}, have also been proposed to reveal the nature of the accelerating universe. Another theory that has been introduced to provide cosmic acceleration is the modified gravity theory such as the \textit{Brans-Dicke theory}  \cite{sen2001late}, \textit{Massive Gravity} \cite{gong2013cosmology},  $f(\mathcal{R})$ \textit{Gravity} \cite{mukherjee2014acceleration} and \textit{Nonlinear Electrodynamics} (NLED) \cite{kruglov2015universe}. All of these models are based on the conservation law of the energy-momentum tensor (EMT).

Among the modified gravity theories that have been considered is the theory proposed by Rastall in which the conservation law of EMT, $\nabla_\mu T^{\mu\nu} =0$ where $\nabla_\mu$ is the covariant derivative, in the curved space-times is replaced by the EMT non-conservation equation with the form $\nabla_\mu T^{\mu\nu} = \nabla^\nu\lambda R$ where $\lambda$ is a constant and $R$ is the spacetime Ricci scalar \cite{rastall1972generalization} such that we recover the theory of GR for $\lambda =0$ and the conservation law of EMT  in flat spacetime.  However, in its development, some dissents arose in response to the Rastall modifications. Some studies conclude that Rastall gravitational theory is equivalent to GR \cite{lindblom1982criticism,visser2018rastall}. However, according to \cite{darabi2018einstein}, the Rastall theory is not equivalent to GR and it is an open theory compared to GR so this theory has the opportunity to answer the problems that have not been explained by the standard cosmological models. In response to this issue, several studies have been conducted to explore the equivalence and non-equivalence between Rastall theory and GR. At the background level, a perfect fluid can be mapped from the Rastall to the GR framework and vice versa \cite{chagoya2023cosmological}. Nevertheless, it should be noted that a fluid described in one theory may not appear identical when viewed from the perspective of the other theory. For example, cold dark matter in the Rastall frame may correspond to warm or hot dark matter in GR. Additionally, at the background level, Rastall cosmology shows no deviation from the $\Lambda$CDM model at late times, regardless of the value of the Rastall parameter \cite{batista2012rastall}. However, the differences seem to appear at the perturbation level. For instance, the dark energy may agglomerates \cite{batista2012rastall} and the growth index of matter perturbations strongly depends on the Rastall parameter in the linear regime \cite{khyllep2019linear}. Furthermore, the deviations between both models become more significant in the non-linear regimes \cite{ziaie2021effects}. Hence, further analysis and precise observational data are needed to determine which model is favored by observations. Phenomenologically, the violation of the conservation law of EMT occurs during the particle creation process in cosmology \cite{gibbons1993cosmological,parker1971quantized}. Then, Rastall theory can be considered as a classical formulation of quantum phenomena that occur in the cosmological scale \cite{silva2013bouncing}. In addition, dark energy may arise as the result of the violation of the EMT's conservation law \cite{josset2017dark} which is a fundamental assumption used in the Rastall theory. The application of this Rastall theory and its generalization to cosmology enable us to have an accelerating universe model \cite{shahidi2021cosmological,capone2010possibility,moradpour2017generalization}. See also, for instance, \cite{lin2020cosmic,llibre2020qualitative,batista2012rastall} for interesting properties of Rastall cosmology. The Rastall theory has also been used to answer another problem from the $\Lambda$CDM model, namely, $H_0$ tension. However, the modification of Friedmann's equation by the non-conservation equation, Rastall-$\Lambda$CDM, leads to a value of $H_0$ which is not much different from that obtained by the $\Lambda$CDM model \cite{akarsu2020rastall}. Fortunately, Rastall-$f(\mathcal{R})$ theory enables us to obtain a $H_0$ value that is closer to the observed value by choosing a suitable form of $f(\mathcal{R})$ function \cite{shahidi2021cosmological}. 

Other intriguing extensions of cosmological models concern the existence of a universe with higher dimensional spacetime, $d > 4$. The emergence of such cosmological models began with the successful unification of electromagnetic and gravitational interactions by Kaluza-Klein in a five-dimensional spacetime \cite{kaluza1921unitatsproblem, klein1987quantum}, where the extra dimensions were compactified on a very small $S^1$ topology. Another motivation for modeling a higher-dimensional universe is the string theory, which requires additional dimensions to achieve consistency \cite{green1984anomaly,witten1995string}. Cosmological models with extra dimensions, beyond the 3+1 dimensions, have been extensively studied and have drawn considerable attention, particularly the Kaluza-Klein model, which aims to explain the inflationary phenomena of the early universe \cite{shafi1983cosmology, abbott1984kaluza, abbott1985kaluza}. From an inflationary perspective, while the usual 3-space undergoes expansion, the extra dimensions experience compactification \cite{shafi1983cosmology} with the speculative natural mechanism \cite{wetterich1985kaluza} in which a scalar field with an exponential potential in the context of four-dimensional theory, corresponding to the variation in the internal space's length scale, acts as a "brake" on the extra dimensions to prevent them from expanding. On the other hand, the cosmological model proposed in \cite{abbott1984kaluza} explains the inflationary mechanism in three-dimensional space occurring when the internal space reaches its maximum size at a specific time and then collapses to its minimum size. This procedure can result in an incredibly large entropy of approximately $\sim 10^{86}$, as discussed in \cite{alvarez1983entropy}.  The necessary conditions for these phenomena, along with the resolution of the horizon and flatness problems, are met when the universe possesses extra dimensions, roughly around $\sim 40$ \cite{abbott1984kaluza}. Besides the Kaluza-Klein cosmology, the compactification of extra dimensions is also predicted by string cosmology \cite{alvarez1985superstring}, see Ref. \cite{lidsey2000superstring} for review. 

In this paper, we study Rastall gravity using a different approach from the previous works mentioned above. We consider higher dimensional Rastall gravity coupled with a canonical scalar field $\phi$ called Klein-Gordon-Rastall (KGR) theory with exponential scalar potential and then study its cosmological consequences using dynamical system analysis. The biggest motivation of our study lies in the lack of discussion of KGR cosmology both in the early and late-time eras in a unified picture. We expect that KGR cosmology might provide a good description of both eras. Here, we perform the calculation of the KGR model on higher dimensional spatially flat Friedmann-Lema$\hat{\mathrm{i}}$tre-Robertson-Walker (FLRW) spacetimes to study its effects on the existence of critical points and their stability. Since we are also studying the inflationary era, considering a higher dimensional spacetime might be loosely reasonable, as it is relevant mainly during inflation \cite{shafi1983cosmology,abbott1984kaluza,abbott1985kaluza,wetterich1985kaluza,alvarez1983entropy,alvarez1985superstring,lidsey2000superstring}. Also, we investigate the cosmological consequences of our model in the absence and the presence of the cosmological constant term. Our work refers to \cite{halliwell1987scalar,wands1993exponential,copeland1998exponential,leon2013cosmological,falah2021higher} who also study the cosmological properties of other theories of gravity using dynamical system analysis. Our model has five critical points related to the cosmological era. We obtain some stable and unstable nodes corresponding to the late-time and early-time attractors, respectively, on scalar field-matter-dominated era called CP4. The description of the early-time universe (unstable nodes and accelerated) obtained from this point leads to the power-law inflation mechanism. We also derive the exact solution of power-law inflation in the KGR framework. A description of the late-time era is given by CP4 for $\Lambda = 0$ case which allows us to have stable solutions and an accelerating universe in KGR theory when non-relativistic matter exists as subdominant components. This feature cannot be found in the standard GR, $\gamma = 0$. In the present model, the accelerating universe can be explained by the nontrivial interplay between the scalar field having an exponential potential and the non-conservative EMT of the non-relativistic matter field in curved spacetime. They play a role as dark energy. However, this scenario cannot be found in CP4 for $\Lambda\neq 0$ case due to its none existence of the stable nodes. 

The structure of this paper is as follows. In Section \eqref{sec:Rastall} we derive the Friedmann equations in KGR theory. Then, we present a way to transform the Friedmann equations for $\Lambda = 0$ case in KGR theory into a set of autonomous equations and perform the dynamical analysis around the critical points in Section \eqref{sec:dynamicalLambda}. The same procedure is applied to the case $\Lambda\neq 0$ in Section \eqref{sec:Lambdaneq}. We construct the local-global existence and the uniqueness of autonomous equations for both $\Lambda = 0$ and $\Lambda \neq 0$ in Section \eqref{sec:LocGlobExis}. Next, in Section \eqref{sec:cosmologicalmodel} we discuss the cosmological implications on both inflationary and late-time accelerating universe and we track the cosmological sequences of KGR theory. Finally, we present our conclusions in Section \eqref{sec:conclusion}. The detailed calculations of the stability conditions of the critical points for $\Lambda = 0$ and $\Lambda \neq 0$ cases are given in Appendix \eqref{sec:appA} and Appendix \eqref{sec:appB}, respectively.


\section{Spatially Flat Cosmology in KGR Theory } 
\label{sec:Rastall}

In this section, we consider a homogeneous and isotropic higher dimensional cosmological model in KGR gravitational theory. Particularly, we use the ansatz metric in the rest of the paper called the spatially flat FLRW metric
\begin{eqnarray}
 \label{eq:flrw}
    ds^2 = - N^2 (t)\, dt^2 + a^2(t) \delta_{ij}\, dx^i dx^j\, ,
\end{eqnarray}
where $\delta_{ij}$ is $\delta$-Kronecker and the indices $i,j=1,2,3,...,(d-1)$ label the  spatial components. From the metric \eqref{eq:flrw}, we obtain the Ricci tensor components 
\begin{align}
   \begin{split}
        R_{00} &= - (d-1) \Big(\frac{\ddot{a}}{a}-\frac{\dot{a}\dot{N}}{aN} \Big)\, ,\qquad R_{0i} = R_{i0}=0 \, ,\\
        R_{ij} &= \delta_{ij} \Big[\frac{a\ddot{a}}{N^2} + (d - 2) \frac{\dot{a}^2}{N^2} - \frac{a\dot{a}\dot{N}} {N^3}\Big]\,, \label{eq:riccitensor}
   \end{split} 
\end{align}
which follows that we have the Ricci scalar 
\begin{eqnarray}
    R = \frac{(d-1)}{N} \bigg[ \frac{2\ddot{a}}{aN} -  \frac{2\dot{a} \dot{N}}{a N ^2} + (d - 2) \frac{\dot{a} ^2}{a ^2 N}\bigg] \, . \label{eq:ricciscalar}
\end{eqnarray}
Let us focus on the EMT in our model. The first part of the EMT is the perfect fluid EMT whose form is given by
\begin{eqnarray}
    T^{(\text{m})}_{\mu\nu} = (\rho_\text{m} + p_\text{m}) u_\mu u_\nu + p_\text{m} g_{\mu\nu}\, , \label{eq:emtmatter}
\end{eqnarray}
where the indices $\mu,\nu = 0, 1,2,3,...,(d-1)$ are the spacetime indices,  $u_{\mu}=-N\delta^0_\mu$,  $p_\text{m} = w_\text{m}\, \rho_\text{m}$, and $ w _\text{m} \in \lR$  which contains dust-like, radiation-like and vacuum-like for $d > 4$.  The second part of the energy-momentum tensor is coming from a canonical scalar field $\phi$, namely, 
\begin{eqnarray}
    T ^{(\phi)} _{\mu\nu} = \partial_\mu \phi\, \partial_\nu \phi - \frac{1}{2} g_{\mu\nu} \, \partial_\alpha \phi\, \partial ^\alpha \phi - g_{\mu\nu} V(\phi)\, , \label{eq:emtscalarfield}
\end{eqnarray}
where $V(\phi)$ is a scalar potential. 

In the Rastall theory, we take the assumption that the divergence of the perfect fluid EMT satisfies 
\begin{eqnarray}
 \nabla^\nu T^{(\text{m})} _{\mu\nu} = \lambda \nabla _{\mu} R\, ,   \label{eq:rastallco}
\end{eqnarray}
where $\lambda$ is a constant Rastall parameter  and $R$ is a Ricci scalar  given by (\ref{eq:ricciscalar}), whereas the scalar field EMT still satisfies the conservation law
\begin{eqnarray}
    \nabla^\nu T^{(\phi)}_{\mu\nu} = 0\, .
    \label{eq:continuityemtscalarfield}
\end{eqnarray}
By adding the equations \eqref{eq:rastallco} and \eqref{eq:continuityemtscalarfield} 
\begin{eqnarray}
    \nabla^\mu \big( T^{(\text{m})}_{\mu\nu} + T^{(\phi)}_{\mu\nu} - \lambda g_{\mu\nu} R\big) = 0\, ,
    \label{eq:nonconsEMT}
\end{eqnarray}
and using the contracted Bianchi identity with cosmological constant
\begin{eqnarray}
    \nabla^\mu \Big( R_{\mu\nu} -\frac{1}{2} g_{\mu\nu}R + \Lambda g_{\mu\nu}\Big)=0\, ,
\end{eqnarray}
then the assumption \eqref{eq:nonconsEMT} is consistent with the following field equations
\begin{eqnarray}
    R_{\mu\nu} + \left(\gamma-\frac{1}{2}\right) g_{\mu\nu}R = \kappa\left(T^{(\text{m})}_{\mu\nu} + T^{(\phi)}_{\mu\nu} - {\tilde{\Lambda}} g_{\mu\nu}\right)\, ,\label{eq:fieldeqn}
\end{eqnarray}
which is commonly known as the Rastall field equations with $\gamma\equiv\kappa\lambda$ and $\tilde{\Lambda}\equiv\frac{\Lambda}{\kappa}$. The quantity $R_{\mu\nu}$ is the Ricci tensor whose components are given by  \eqref{eq:riccitensor}. Some comments are in order. First, we recover the standard higher dimensional GR by taking  $\lambda=0$. Second, the conservation law of the scalar field EMT \eqref{eq:continuityemtscalarfield} is ensured by the scalar field equation of motions
\begin{eqnarray}
   g^{\mu\nu}\nabla_\mu \partial_\nu \phi = - \frac{\partial V}{\partial\phi} ~ .\label{eq:scalarfieldeom}
\end{eqnarray}

Now, let us focus on the Rastall field equation \eqref{eq:fieldeqn} on the metric \eqref{eq:flrw}. Since we are working on the homogeneous metric, the scalar field would be time dependent only. Indeed, it is consistent with $0i$ components of Eq. \eqref{eq:fieldeqn}. By substituting equation \eqref{eq:riccitensor},\eqref{eq:ricciscalar},\eqref{eq:emtmatter}, and \eqref{eq:emtscalarfield} into equation \eqref{eq:fieldeqn}, one obtains the modified first Friedmann equation \\
\begin{align}
\frac{(d-1) ( d - 2) }{2} H ^2 = \kappa (\rho_\text{m} + \rho_\text{KGR} + \rho_{\Lambda})\, , \label{eq:friedmann}
\end{align}
 where
\begin{eqnarray}
   \rho_\text{KGR} \equiv \bigg[1 - \frac{4 \gamma(d-1)}{(d-2)} \bigg] \frac{\dot{\phi}^2}{2 N^2} + V(\phi) - \frac{2\gamma(d-1)}{(d-2)}(1 + w_\text{m}) \, \rho_\text{m}\, + \frac{d \, (d-1) \gamma H ^2}{\kappa} \, ,
\end{eqnarray}
and
\begin{eqnarray}
    \rho_{\Lambda} \equiv \tilde{\Lambda} \, .
\end{eqnarray}
One also obtains the modified second Friedmann equation
\begin{align}
    [d - 2 -  2\gamma (d-1)] \frac{\dot{H}}{N} + \frac{1}{2} (d-1) (d - 2 - 2\gamma d) H^2  = - \kappa \left( p_\text{m} + \frac{\dot{\phi}^2}{2N^2} - V(\phi) -\tilde{\Lambda}\right) \, .
    \label{eq:raychauduri}
\end{align}
Here, we have defined the Hubble parameter  
\begin{eqnarray}
    H = \frac{\dot{a}}{aN}\, . 
    \label{eq:hubble}
\end{eqnarray}
The time component of  \eqref{eq:rastallco} implies the fluid equation of higher dimensional Rastall theory 
\begin{eqnarray}
    \dot{\rho}_\text{m} + (d-1)NH(\rho_\text{m}+p_\text{m}) = -\frac{2 \lambda(d-1)}{N}\Bigg(\ddot{H}-\frac{\dot{N}\dot{H}}{N}+dNH\dot{H}\Bigg)\, , 
    \label{A}
\end{eqnarray}
in which we have
\begin{eqnarray}
\label{eq:ray}
    \frac{\dot{H}}{N} = -\frac{\kappa}{(d-2)} \Big[\rho_\text{m}+p_\text{m} + \frac{\dot{\phi}^2}{N^2}\Big] \, ,
\end{eqnarray}
and
\begin{align}
  \ddot{H} = - \frac{\kappa}{(d-2)} \Bigg[\dot{N}(\rho_m+p_m) + N(\dot{\rho}_m+\dot{p}_m)-\frac{\dot{N}}{N^2}\dot{\phi}^2 + \frac{2\dot{\phi}\ddot{\phi}}{N}\Bigg] 
 \, . \label{eq:acc}
\end{align}
By substituting Eq. \eqref{eq:ray} into Eq. \eqref{eq:raychauduri}, we can rewrite the second Friedmann equation in the following form
\begin{align}
    (d - 2) \frac{\dot{H}}{N} + \frac{1}{2} (d-1) (d - 2) H^2 = - \kappa \left( p_\text{m} + p_\text{KGR} + p_{\Lambda}\right) \, ,
    \label{eq:friedmannII}
\end{align}
\begin{eqnarray}
p_\text{KGR} \equiv  \bigg[1 + \frac{4 \gamma (d-1)}{(d-2)} \bigg] \frac{\dot{\phi}^2}{2 N^2} - V(\phi) + \frac{2\gamma (d - 1)}{(d-2)}(1 + w_\text{m}) \, \rho_\text{m} - \frac{d \, (d-1) \gamma H ^2}{\kappa} \, ,
\end{eqnarray}
and
\begin{eqnarray}
    p_{\Lambda} \equiv - \tilde{\Lambda} \, .
\end{eqnarray}
From Eq. \eqref{eq:friedmann} and Eq. \eqref{eq:friedmannII}, we can define the equation of state parameters as follows
\begin{align}
    w_\text{m}  \equiv &  \frac{p_\text{m}}{\rho_\text{m}} \, , \nonumber \\ 
    w_\text{KGR} \equiv & \frac{p_\text{KGR}}{\rho_\text{KGR}} = \frac{\left[1 + \frac{4 \gamma (d - 1)}{(d-2)} \right] \frac{\dot{\phi}^2}{2 N^2} - V(\phi) + \frac{2\gamma (d - 1)}{(d-2)}(1 + w_\text{m}) \, \rho_\text{m}- \frac{d \, (d-1) \gamma H ^2}{\kappa}}{\left[1 - \frac{4 \gamma(d-1)}{(d-2)} \right] \frac{\dot{\phi}^2}{2 N^2} + V(\phi) - \frac{2\gamma(d-1)}{(d-2)}(1 + w_\text{m}) \, \rho_\text{m}+\frac{d \, (d-1) \gamma H ^2}{\kappa}} \, , \\
    w_{\Lambda} \equiv & \frac{p_{\Lambda}}{\rho_{\Lambda}} = -1 \, \nonumber.
\end{align}
In addition, the effective equation of state parameter can be defined as \cite{copeland2006dynamics}
\begin{align}
w_\text{eff} \equiv & \frac{p_\text{m} + p_\text{KGR} + p_{\Lambda} }{\rho_\text{m} + \rho_\text{KGR} + \rho_{\Lambda}} \nonumber \\
 =& \frac{w_\text{m} \, \rho_\text{m} + \left[1 + \frac{4 \gamma (d - 1)}{(d-2)} \right] \frac{\dot{\phi}^2}{2 N^2} - V(\phi) + \frac{2\gamma (d - 1)}{(d-2)}(1 + w_\text{m}) \, \rho_\text{m}- \frac{d \, (d-1) \gamma H ^2 }{\kappa} - \tilde{\Lambda}}{\rho_\text{m} + \left[1 - \frac{4 \gamma(d-1)}{(d-2)} \right] \frac{\dot{\phi}^2}{2 N^2} + V(\phi) - \frac{2\gamma(d-1)}{(d-2)}(1 + w_\text{m}) \, \rho_\text{m}+\frac{d \, (d-1) \gamma H ^2 }{\kappa}+ \tilde{\Lambda}} \, , 
\end{align}
It is worth mentioning that eq. \eqref{A} is consistent with the scalar field equation of motions \eqref{eq:scalarfieldeom} whose form on the metric \eqref{eq:flrw} is given by
\begin{equation}
\label{eq:scalareomflrw}
 \ddot{\phi} + \Bigg[(d-1)NH - \frac{\dot{N}}{N}\Bigg]\dot{\phi} + N^2V,_{\phi} = 0\, .   
\end{equation}
Analysing on the spatial components of \eqref{eq:rastallco}, we find $p_\text{m}(t)$ and $\rho_\text{m}(t)$.
Hence, we have a model of higher dimensional Rastall cosmology with homogeneous perfect fluid. Such a model has been studied in \cite{chatterjee1990homogeneous} where $w_\text{m}=0$ for dust, $w_\text{m}=\frac{1}{(d-1)}$ for radiation and $w_\text{m}=-1$ for vacuum.


\section{Dynamical System Analysis: $\Lambda = 0$ Case} \label{sec:dynamicalLambda}

In this section, we will investigate the behavior of Rastall cosmology using the dynamical system in the absence of the cosmological constant term, namely $\Lambda = 0$. To transform our cosmological equations from the previous section to \textit{autonomous} equations, we have to choose a particular form of the scalar field potential,
\begin{equation}\label{eq:Vpotexp}
    V=V_0\exp{(-\sqrt{\kappa}\lambda_V\phi)} ~ .
\end{equation}
Gravity-coupled scalar field theories with an exponential potential have been widely considered in various physical theories such as Kaluza-Klein cosmology \cite{wetterich1985kaluza}, the Salam-Sezgin model (supergravity with ${\mathcal N} = 2$ coupled to matter in six dimensions) \cite{salam1984chiral}, and the exponential potential also provides a solution for power-law inflation \cite{lucchin1985power} where the scale factor expands according to $a(t) \propto t^\ell $ with $\ell > 1$. 

This type of potential has also been studied by several authors \cite{halliwell1987scalar, wands1993exponential, copeland1998exponential} in the context of dynamical systems focusing on either the inflationary phase or the late-time universe. Also, there is an attempt to investigate the inflationary and the dark energy eras in the massive gravity framework \cite{falah2021higher}. This inspires us to investigate these two crucial eras in the context of Rastall gravity.

First, let us define three autonomous variables as follows
\begin{eqnarray}
\label{eq:autoxm}
    x_\text{m} = \sqrt{\frac{2\kappa \rho_m}{(d-1)(d-2)H^2}}\, ,
\end{eqnarray}
\begin{eqnarray}
\label{eq:autoxp}
    x_{\phi} = \sqrt{\frac{\kappa \dot{\phi}^2}{(d-1)(d-2)N^2H^2}}\, ,
\end{eqnarray}
\begin{eqnarray}
\label{eq:autoxv}
    x_V = \sqrt{\frac{2\kappa V(\phi)}{(d-1)(d-2)H^2}}\, .
\end{eqnarray}
We have the equations of motion correspond to the autonomous variables $x_{\phi}$ and $x_V$ by differentiating them with respect to $\ln{a}$ and using \eqref{eq:raychauduri} and \eqref{eq:scalareomflrw}
\begin{align}
\label{eq:xphi}
   \frac{2}{(d-1)N} x_{\phi}^{'} =& \left(\frac{\Big[1 - \frac{2\gamma d}{(d-2)}\Big](1+w_\text{m})}{\Big[1-\frac{2\gamma(d-1)}{(d-2)}(1+w_\text{m})\Big]}-2\right)x_{\phi} +  \lambda_V \sqrt{\frac{(d-2)}{(d-1)}} x_V^2 \nonumber \\
   &+\frac{(1-w_\text{m})}{\Big[1-\frac{2\gamma(d-1)}{(d-2)}(1+w_\text{m})\Big]} x_{\phi}^3 - \frac{(1+w_\text{m})}{\Big[1-\frac{2\gamma(d-1)}{d-2}(1+w_\text{m})\Big]} x_{\phi} x_V^2 \, ,
\end{align}
\vspace{0.4cm}
\begin{align}
\label{eq:xv}
    \frac{2}{(d-1)N} x_V^{'} =& \frac{\Big[1-\frac{2\gamma d}{(d-2)}\Big] (1+w_\text{m})}{\Big[1-\frac{2\gamma(d-1)}{(d-2)}(1+w_\text{m})\Big]} x_V - \lambda_V \sqrt{\frac{(d-2)}{(d-1)}} x_{\phi} x_V \nonumber\\
    &+\frac{(1-w_\text{m})}{\Big[1-\frac{2\gamma(d-1)}{(d-2)}(1+w_\text{m})\Big]} x_V x_{\phi}^2 - \frac{(1+w_\text{m})}{\Big[1-\frac{2\gamma(d-1)}{d-2}(1+w_\text{m})\Big]} x_V^3 \, .
\end{align}
while, by using \eqref{eq:autoxm},\eqref{eq:autoxp} and \eqref{eq:autoxv}, the Friedmann equation \eqref{eq:friedmann} can be rewritten as 
\begin{equation}
\label{constraineq}
    \left[1-\frac{2\gamma(d-1)}{(d-2)}(1+w_m)\right]x_m^2+ \left[1-\frac{4\gamma(d-1)}{(d-2)}\right]x_{\phi}^2+x_V^2=1-\frac{2\gamma d}{(d-2)}\, .
\end{equation}
which can be viewed as a constraint equation. Following \cite{bohmercosmological}, let us introduce the density parameters 
\begin{eqnarray}\label{eq:Omegam}
    \Omega_\text{m} \equiv \frac{2\kappa\rho_\text{m}}{(d-1)(d-2)H^2} = x_\text{m}^2\, ,
\end{eqnarray}
\begin{align} \label{eq:Omegaphi}
    \Omega_\text{KGR} \equiv  \frac{2\kappa\rho_\text{KGR}}{(d-1)(d-2)H^2} = & \frac{1}{\left[1 - \frac{2\gamma (d-1)}{(d-2) } (1 + w_\text{m})\right]} \Bigg( \Bigg[1 -\frac{4\gamma(d-1)}{(d-2)}\Bigg] x_{\phi}^2 + x_V^2  \\ \nonumber
   & + \frac{2\gamma d}{(d-2)} - \frac{2\gamma (d-1)}{(d-2)} (1 + w_\text{m}) \Bigg) \, .
\end{align}
Then, we can rewrite the equation of state parameter of the KGR term
\begin{eqnarray} \label{eq:wR}
	w_\text{KGR}  = \frac{\left[1 + \frac{4 \gamma (d - 1)}{(d-2)} \right] x_\phi^2 - x_V^2 + \frac{2\gamma (d - 1)}{(d-2)}(1 + w_\text{m}) \, x_\text{m}^2- \frac{2 \gamma d}{(d-2) }}{\left[1 - \frac{4 \gamma(d-1)}{(d-2)} \right] x_\phi^2 + x_V^2  - \frac{2\gamma(d-1)}{(d-2)}(1 + w_\text{m}) \, x_\text{m}^2+ \frac{2 \gamma d}{(d-2)}}  \, ,
\end{eqnarray}
and the effective equation of state parameter
\begin{eqnarray} \label{eq:weff}
w_\text{eff}  = \frac{w_\text{m} x_\text{m}^2 + \left[1 + \frac{4 \gamma (d - 1)}{(d-2)} \right] x_\phi^2 - x_V^2 + \frac{2\gamma (d - 1)}{(d-2)}(1 + w_\text{m}) \, x_\text{m}^2- \frac{2 \gamma d}{(d-2) }}{x_\text{m}^2  + \left[1 - \frac{4 \gamma(d-1)}{(d-2)} \right] x_\phi^2 + x_V^2  - \frac{2\gamma(d-1)}{(d-2)}(1 + w_\text{m}) \, x_\text{m}^2+ \frac{2 \gamma d}{(d-2)}}  \, .
\end{eqnarray} 
It should be noted that if $\gamma = 0$, then $\Omega_\text{KGR} = x_{\phi}^2 + x_V^2$ and $w_{\text{KGR}} = (x_{\phi}^2 - x_V^2) / (x_{\phi}^2 + x_V^2)$. These equations are equivalent to the ones used in the standard scalar field model within the general relativity framework. The equation \eqref{constraineq} can be cast into the form 
\begin{eqnarray}
      \Omega_\text{m} + \Omega_\text{KGR} = 1\, .
\end{eqnarray}
It is clear from \eqref{eq:Omegam} that  $\Omega_\text{m} \geq 0$ implying the restriction $ \Omega_\text{KGR} \leq 1$. For example, if we consider $\Lambda$CDM model in which we have  $\Omega_\text{m} \approx 0.3 $, then  $ \Omega_\text{KGR} \approx 0.7 $ \cite{abbott2019first}. In this sense, the energy density of the cosmological constant is replaced by $ \Omega_\text{KGR}$ and it might play a role as the dark energy as long as $w_\text{KGR} < -\frac{1}{(d-1)}$. Thus, in our model, dark energy can be seen as a consequence of the existence of a scalar field and the ability of the geometry to couple the matter field which drives our universe to expand at an accelerating rate. To see this, we first define  the deceleration parameter
\begin{equation}
    q \equiv -1 -\frac{\dot{H}}{N H^2} = -1 + \frac{(d-1)}{2\left[1-\frac{2\gamma(d-1)}{(d-2)}\right]} \Big[1-\frac{2\gamma d}{(d-2)}+w_\text{m} x_\text{m}^2 + x_{\phi}^2 - x_V^2\Big]\, ,
\end{equation}
where $x_m^2$ can be obtained from the constraint equation \eqref{constraineq}. An indication of an accelerating universe is given by $q<0$. Note that unlike in the GR case, here  $ \Omega_\text{KGR}$ is not necessarily greater than zero, this depends on the density parameter $\rho_\text{m}$. For example, if $\Omega_\text{m} > 1$, we have  $ \Omega_\text{KGR} < 0$ in contrast to the case when $\Omega_\text{m} < 1$ where we have $ \Omega_\text{KGR} > 0$. 

In our model, we have five critical points related to the cosmological eras. We denote each critical point ($x_{\phi,c}$, $x_{V,c}$) as CP$i$, where $i=1,2,3,4,5$. The variables $x_{\phi,c}$ and $x_{V,c}$ represent the critical values of $x_{\phi}$ and $x_{V}$, respectively. The critical points, leading to the right-hand side of Eq. \eqref{eq:xphi} and Eq. \eqref{eq:xv} to become zero simultaneously for $\Lambda = 0$ case, are given as follows:
%
\begin{itemize}
	\item CP1 $(0,0)$:\\
    CP1 always exists for any $w_m$, $\lambda_V$. The properties of CP1 can be described by the following parameters
	\begin{align} \label{eq:qcp1}
	\begin{split}
	 \Omega_\text{KGR} &= \frac{\frac{2\gamma d}{(d-2)} - \frac{2\gamma(d-1)}{(d-2)}(1+w_\text{m})}{\left[1 - \frac{2\gamma(d-1)}{(d-2)}(1+w_\text{m})\right]}\, ,\\ 
	  q &= -1 + \frac{(d-1)\left[1-\frac{2\gamma d}{(d-2)}\right](1 + w_\text{m})}{2\left[1-\frac{2\gamma(d-1)}{(d-2)}(1 + w_\text{m})\right]}\, .
	\end{split}
	\end{align}
 The value of the equation of state parameter $w_\text{KGR}$ for the scalar field is undefined when $\gamma = 0$. Otherwise, $w_\text{KGR} =  -1$ such that it mimics dark energy which behaves like a cosmological constant. Since the energy density of the scalar field can be neglected, the role of dark energy is played by the non-conservative EMT of a perfect fluid. The existence and the stability properties of the CP1 point, which describes the universe during the matter-dominated era with coupling between the matter field and geometry via the Rastall parameter, are illustrated in Fig. \ref{fig:cp1}. In contrast to the case in general relativity, the stability of this point is determined by the parameters $w_\text{m}$ and $\gamma$. There are only two possible outcomes from this point: a past-time attractor (unstable node) and a saddle-node. The instability of this point implies that, even though it leads to an accelerating universe, it cannot represent the late-time universe.\\
 \begin{figure}[h]
 	\centering
 		\begin{subfigure}[t]{0.32\textwidth}
 			\includegraphics[width=1\textwidth]{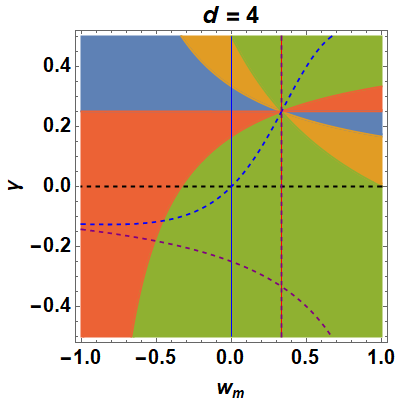}
 			\caption{$ d=4$}
 		\end{subfigure}   
 		\hfill
 		\begin{subfigure}[t]{0.32\textwidth}
 			\includegraphics[width=1\textwidth]{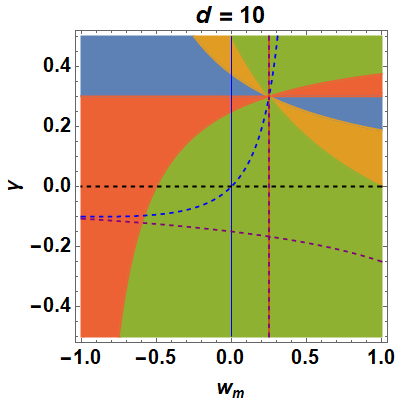}
 			\caption{$d=10$}
 		\end{subfigure} 
 		\hfill
 		\begin{subfigure}[t]{0.32\textwidth}
 			\includegraphics[width=1\textwidth]{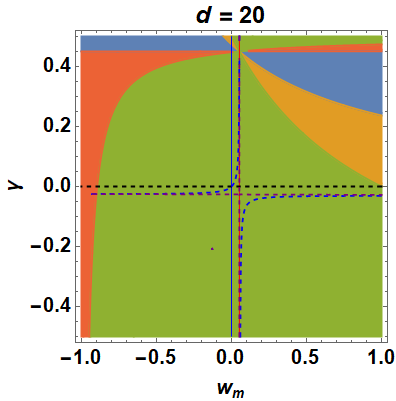}
 			\caption{$d=20$}
 		\end{subfigure}   
 		\vfill
 		\begin{subfigure}[normla]{0.25\textwidth}
 			\includegraphics[scale=0.8]{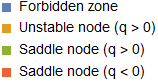}
 		\end{subfigure}        
 		\caption{The bifurcation diagrams show the existence and the stability conditions of CP1 by plotting $\gamma$ as a function of $w_\text{m}$ for a constant value of (a) $d=4$, (b) $d=10$ and (c) $d=20$ for both $\Lambda=0$ and $\Lambda\neq 0$ cases. These figures show that in a radiation-like dominated universe (red line), we always obtain saddle nodes and a non-accelerating universe for arbitrary $d\geq 4$ both in GR and Rastall framework. In a dust-like matter universe (blue line), we have saddle nodes and an accelerating universe, saddle nodes and a non-accelerating universe, or unstable nodes and a non-accelerating universe in Rastall sector $\gamma \neq 0$, whilst at GR limit we have only saddle nodes and a non-accelerating universe. From the figures, in the case of $w_\text{m}= 0$, it can be seen that the cosmic acceleration is barely obtained for large $d$ such that it tends to vanish at very large $d$. The blue-dashed and purple-dashed curves denotes $w_\text{eff} = 0$ and $w_\text{eff} = \frac{1}{(d-1)}$, respectively. Note that we have a forbidden zone giving $\Omega_\text{m} < 0$. }
 		\label{fig:cp1}    
 	
 \end{figure}
     %
    \item CP2 $\Bigg(\sqrt{1-\frac{2\gamma(1+w_m)}{(1-w_m)}} , 0\Bigg)$ and CP3 $\Bigg(-\sqrt{1-\frac{2\gamma(1+w_m)}{(1-w_m)}} , 0\Bigg)$:\\
    CP2 and CP3 exist if $w_m\neq 1$ and $1-\frac{2\gamma(1+w_m)}{(1-w_m)}>0$. The properties of CP2 and CP3 can be described by the following parameters
\begin{align}\label{eq:qcp2}
	\Omega_\text{KGR} = \frac{4\gamma + w_\text{m} - 1}{w_\text{m} - 1}\, , \quad
	q =\, d-2\,  .
\end{align}
   The equation of state parameter is
\begin{align}\label{eq:eomcp2}
	w_\text{KGR} = \frac{(1 + 4\gamma)w_\text{m} - 1}{4\gamma + w_\text{m} - 1}\, .
\end{align}
%
 Fig. \ref{fig:cp2} and Fig.  \ref{fig:cp3} display the existence and stability characteristics of CP2 and CP3, respectively. These points correspond to the kinetic-matter-dominated era, which becomes a kinetic-dominated era when $\gamma=0$. Three possible stabilities, namely stable nodes, unstable nodes, and saddle nodes, can be obtained at CP2. Although we have stable nodes, based on the deceleration parameter of CP2 $(q > 0)$, it cannot describe the late-time universe since it is disfavoured by observations. In terms of its existence and the values of $\Omega_\text{KGR}, w_\text{KGR}$, and $q$, CP3 exhibits similar properties as CP2. Only saddle points and unstable nodes exist at this critical point, leading to a non-accelerating expansion of the universe such that it also cannot describe the late-time universe.
 \begin{figure}[h]
 \centering
 	\begin{subfigure}[t]{0.32\textwidth}
 		\includegraphics[width=1\textwidth]{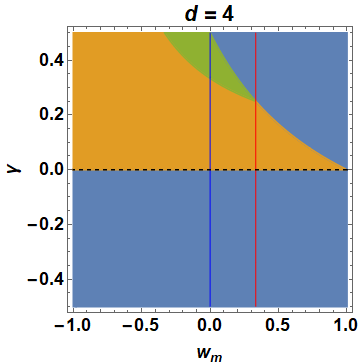}
 		\caption{$ d=4$}
 	\end{subfigure}   
 	\hfill
 	\begin{subfigure}[t]{0.32\textwidth}
 		\includegraphics[width=1\textwidth]{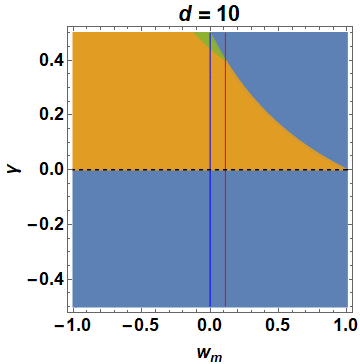}
 		\caption{$d=10$}
 	\end{subfigure} 
 	\hfill
 	\begin{subfigure}[t]{0.32\textwidth}
 		\includegraphics[width=1\textwidth]{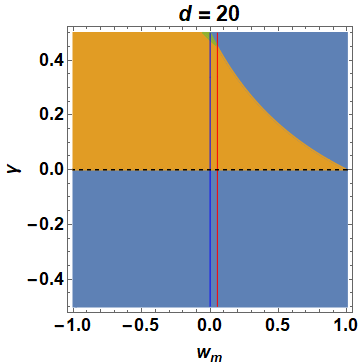}
 		\caption{$d=20$}
 	\end{subfigure}   
 	\vfill
 	\begin{subfigure}[t]{0.6\textwidth}
 	\includegraphics[width=1\textwidth]{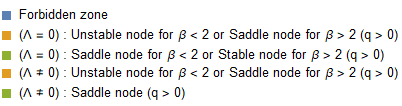}
 	\end{subfigure}      
 		\caption{ The bifurcation diagrams show the existence and the stability conditions of CP2 by plotting $\gamma$ as a function of $w_\text{m}$ for a constant value of (a) $d=4$, (b) $d=10$ and (c) $d=20$. The stability characteristics of this point are given in both $\Lambda=0$ and  $\Lambda\neq 0$ cases which can be seen in the figure legend. In the Rastall sector, these figures show that in a kinetic-radiation-like universe (red line), we obtain either unstable or saddle nodes depending on $\beta$, , where $\beta \equiv \lambda_V \sqrt{\frac{d-2}{d-1}}\sqrt{1-\frac{2\gamma(1+ w_\text{m})}{(1- w_\text{m})}}$, for both $\Lambda=0$ and  $\Lambda\neq 0$ cases. In the kinetic-dust-like matter universe (blue line), we have either saddle or stable nodes depending on parameter $\beta$ for $\Lambda=0$ case whilst it becomes only saddle for $\Lambda\neq 0$ case. At GR limit (black-dashed) we have unstable or saddle nodes depending on $\beta$ for arbitrary $\Lambda$. From the figures, in the case of $w_\text{m} =0$, it can be seen that the stable nodes nearly vanish at very large $d$ in the ranges of $\gamma < \frac{(d-2)}{2d}$.  Note that we have forbidden zones giving either $\Omega_\text{m} < 0$ or $\sqrt{1-\frac{2\gamma(1+w_\text{m})}{(1-w_\text{m})}}<0 $. }
 		\label{fig:cp2}    
 \end{figure}
   %
   \begin{figure}[t]
   	\centering
   	\begin{subfigure}[t]{0.3\textwidth}
   		\includegraphics[width=1\textwidth]{CP2d_4.png}
   		\caption{$ d=4$}
   	\end{subfigure}   
   	\hfill
   	\begin{subfigure}[t]{0.3\textwidth}
   		\includegraphics[width=1\textwidth]{CP2d_10.png}
   		\caption{$d=10$}
   	\end{subfigure} 
   	\hfill
   	\begin{subfigure}[t]{0.3\textwidth}
   		\includegraphics[width=1\textwidth]{CP2d_20.png}
   		\caption{$d=20$}
   	\end{subfigure}   
   	\vfill
   	\begin{subfigure}[t]{0.25\textwidth}
   		\includegraphics[scale=0.85]{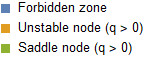}
   	\end{subfigure}              
   	\caption{ The bifurcation diagrams show the existence and the stability conditions of CP3 by plotting $\gamma$ as a function of $w_\text{m}$ for a constant value of (a) $d=4$, (b) $d=10$ and (c) $d=20$ for both $\Lambda=0$ and  $\Lambda\neq 0$ cases. In Rastall sector $(\gamma \neq 0)$, these figures show that in a kinetic-radiation-like universe (red line), we obtain unstable nodes. In a kinetic-dust-like matter universe (blue line), we have either unstable or saddle nodes. At GR limit (black dashed) we have only unstable nodes. Note that we have forbidden zones giving either $\Omega_\text{m} < 0$ or $\sqrt{1-\frac{2\gamma(1+w_\text{m})}{(1-w_\text{m})}}<0 $.}
   	\label{fig:cp3}    
   \end{figure}
\begin{figure*}[!ht]
	\centering
	\begin{subfigure}[t]{0.38\textwidth}
		\includegraphics[width=1\textwidth]{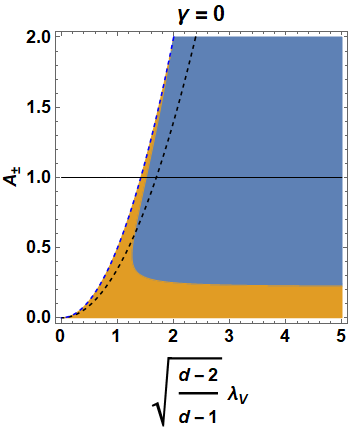}
	\end{subfigure}   \\ 
	\begin{subfigure}[t]{0.45\textwidth}
		\includegraphics[width=0.9\textwidth]{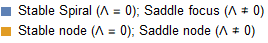}
	\end{subfigure}
	\caption{ The bifurcation diagram showing the existence and the stability conditions of CP4 by plotting $\mathcal{A}_{\pm}$ as a function of $\lambda_V$. It is known as a scaling solution in which a scalar field mimics the dominant matter in the GR limit. Such a case had been considered in \cite{copeland1998exponential} and does not gives late-time cosmic acceleration when dust matter is a subdominant constituent. Note that the black-dashed and the black-thin illustrate the $\Omega_\text{m} \approx 0.3$ and $w_\text{m} = 0$, respectively. Accelerating universe obtained for $\mathcal{A}_{\pm} < \frac{2}{(d-1)}$ correspond to $w_\text{KGR}  < -\frac{1}{(d-1)}$. Otherwise, we have a non-accelerating one. The blue-dashed curve denotes $\Omega_\text{KGR} \approx 1$ which indicates eternal power-law inflation for $\Lambda = 0$ and $q < 0$ due to its late-time attractors behaviour.}
\label{fig:cp4a} 
\end{figure*}
    \item CP4 $\Bigg(\frac{1}{\lambda_V} \sqrt{\frac{d-1}{d-2}} \mathcal{A}_{\pm} , \frac{1}{\lambda_V} \sqrt{\frac{(d-1)(2 - \mathcal{A}_{\pm}) 
    		\mathcal{A}_{\pm}}{(d-2)}}\Bigg)$:\\
    CP4 exists if  $0< \mathcal{A}_{\pm}< 2$. 
    The properties of CP4 can be described by the following parameters 
\begin{align}  \label{eq:Ocp4}
	\Omega_\text{KGR} =  \frac{2(d-1)\mathcal{A}_{\pm}[2\mathcal{A}_{\pm}(d-1)-(d-2)] + 2\gamma(d-2)\lambda_V^2 (d w_\text{m} - w_\text{m} -1)}{(d-2)\lambda_V^2[d-2-2\gamma(d-1)(1+w_\text{m})]} \, ,
\end{align}
\begin{align}\label{eq:qcp4}
	q = \frac{2\mathcal{A}_{\pm}(d-1)^2[1+w_\text{m}-\mathcal{A}_{\pm}]+(d-2)[3+w_\text{m} - 2\gamma (d-2) (1+w_\text{m})]\lambda_V^2}{2\lambda_V^2 [d-2-2\gamma(d-1)(1+w_\text{m})]}\, ,
\end{align}
\begin{align}\label{eq:eomcp4}
w_\text{KGR} =& \frac{\mathcal{A}_{\pm}(d-1)\left\{d - 2 - \mathcal{A}_{\pm}[d -2 - 2\gamma (d - 1) w_\text{m}]\right\}} { \mathcal{A}_{\pm} (d-1) \left\{2\gamma  \mathcal{A}_{\pm} (d - 1) - (d - 2)\right\} + (d - 2)\gamma (d w_\text{m} - w_\text{m} - 1)\lambda_V^2} \, \\ \nonumber
&- \frac{(d - 2)\gamma (d w_\text{m} - w_\text{m} - 1)\lambda_V^2} { \mathcal{A}_{\pm} (d-1) \left\{2\gamma  \mathcal{A}_{\pm} (d - 1) - (d - 2)\right\} + (d - 2)\gamma (d w_\text{m} - w_\text{m} - 1)\lambda_V^2}
\, ,
\end{align}
where the parameter $\mathcal{A}_{\pm}(d,w_\text{m},\gamma,\lambda_V)$ is defined as
\begin{align}
	\mathcal{A}_{\pm} \equiv&  \bigg\{\frac{1+w_\text{m}}{2}+\frac{\lambda_V^2(d-2)}{4(d-1)}\bigg[1-\frac{2\gamma(d-1)(1+w_\text{m})}{(d-2)}\bigg]\bigg\} \nonumber\\
	&\times \left\{1 \pm \sqrt{1- \frac{16\lambda_V^2(d-1)(1 + w_\text{m})(d-2-2\gamma d)}{\Big[2(d-1)(1 + w_\text{m}) + \lambda_V^2 [d-2-2\gamma(d-1)(1+ w_\text{m})]\Big]^2}}\right\} \, ,\label{Aplus}
\end{align}
such that
     %
\begin{eqnarray}
	w_\text{m} = -1 + \frac{\mathcal{A}_{\pm}[(d-2)\lambda_V^2 - 2 (d-1)\mathcal{A}_{\pm}]}{(d-2-2\gamma d)\lambda_V^2 + 2 \mathcal{A}_{\pm} (d-1)(\gamma \lambda_V^2 - 1)}\, .
	\end{eqnarray}
The existence and the stability properties of CP4 can be observed in Fig. \ref{fig:cp4t}. CP4 represents a scalar field-matter-dominated era, in which there exists a fraction between the scalar field and matter energy density. The critical point exists in the range of $0 < \mathcal{A_{\pm}} <2$. At this point, we find late-time attractors (stable nodes), past-time attractors (unstable nodes), and saddle points that depend on the parameters $\mathcal{A_{\pm}}(d,w_\text{m},\lambda_V,\gamma)$. The density parameter for CP4 is given by $\Omega_\text{m} + \Omega_\text{KGR} = 1$, such that $\Omega_\text{KGR} \leq 1$. However, there is a forbidden zone where this constraint is not satisfied, $\Omega_\text{KGR}> 1$ or $\Omega_{\text{m}} < 0$. We find that there are unstable nodes. It provides a good description of the early-time universe associated with the power-law inflation era. In the KGR theory, $\gamma \equiv \kappa \lambda \neq 0$, Eq. \eqref{eq:rastallco} prohibits $\rho_\text{m}$ to be zero. It means that even during inflation, the perfect fluid must be presented. Such models have been widely considered in \cite{barrow1993scalar, banerjee1998power}. Fortunately, we still obtain power-law inflation since the scale factor evolves according to $a(t) \propto t^\ell$ where $\ell >1$. On the other 
hand, the late-time attractor or stable nodes can also be obtained, which enables the universe to undergo accelerated expansion for $\gamma \neq 0$, when dust is the subdominant constituent. These features will be discussed further in Sec. \eqref{sec:cosmologicalmodel}. 
In Fig. \ref{fig:cp4t} we use the black-dashed to show the $\Omega_\text{KGR} \approx 0.7$ $(\Lambda$CDM-like model) curve. We also use the black-thin to show a scalar field-dust-like matter universe $(w_\text{m} = 0)$ curve. Note that uncolored regions denote a forbidden zone leading to $\Omega_\text{m} < 0$. 
\begin{figure*}[hptb!]
	\centering
	\begin{subfigure}[t]{0.32\textwidth}
		\includegraphics[width=1\textwidth]{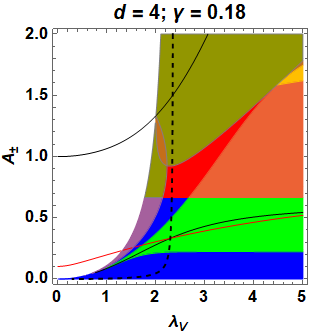}
		\caption{$d=4, \gamma = 0.18$}
	\end{subfigure}   
	\hfill
	\begin{subfigure}[t]{0.32\textwidth}
		\includegraphics[width=1\textwidth]{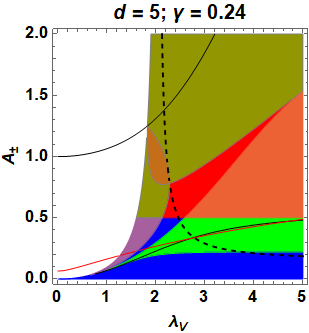}
		\caption{$d=5, \gamma = 0.24$}
	\end{subfigure} 
	\hfill
	\begin{subfigure}[t]{0.32\textwidth}
		\includegraphics[width=1\textwidth]{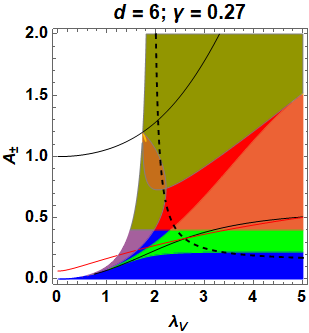}
		\caption{$d=6, \gamma = 0.27$}
	\end{subfigure}
	\vfill
	\begin{subfigure}[t]{0.32\textwidth}
		\includegraphics[width=1\textwidth]{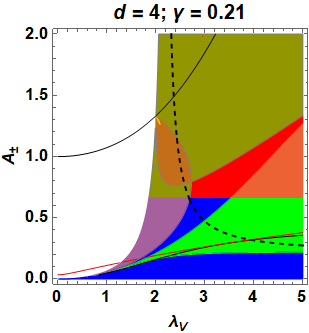}
		\caption{$d=4, \gamma = 0.21$}
	\end{subfigure}   
	\hfill
	\begin{subfigure}[t]{0.32\textwidth}
		\includegraphics[width=1\textwidth]{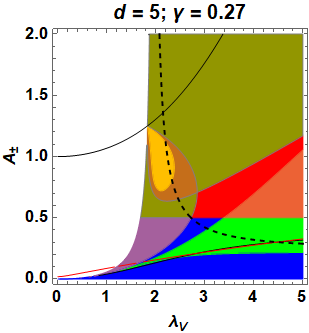}
		\caption{$d=5, \gamma = 0.27$}
	\end{subfigure} 
	\hfill
	\begin{subfigure}[t]{0.32\textwidth}
		\includegraphics[width=1\textwidth]{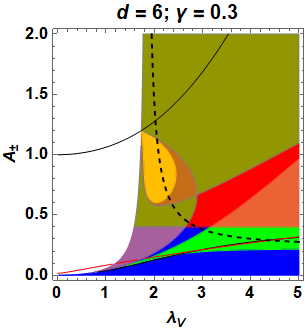}
		\caption{$d=6, \gamma = 0.3$}
	\end{subfigure}
	\vfill       
	\begin{subfigure}[t]{0.3\textwidth}
		\includegraphics[width=1\textwidth]{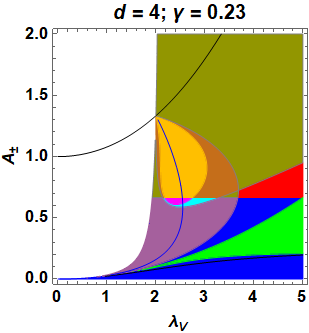}
		\caption{$d=4, \gamma = 0.23$}
    \label{subfig:g}
	\end{subfigure}	
	\hfill       
	\begin{subfigure}[t]{0.3\textwidth}
		\includegraphics[width=1\textwidth]{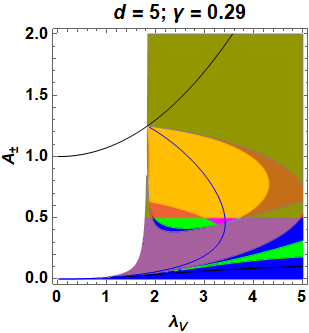}
		\caption{$d=5, \gamma = 0.29$}
    \label{subfig:h}
	\end{subfigure}	
	\hfill
	\begin{subfigure}[t]{0.3\textwidth}
		\includegraphics[width=1\textwidth]{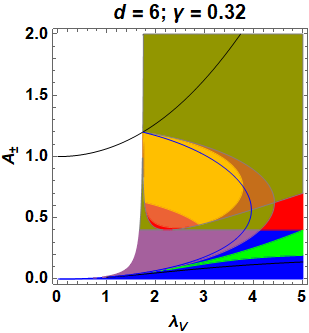}
		\caption{$d=6, \gamma = 0.32$}
    \label{subfig:i}
	\end{subfigure}	
	\vfill 
	\begin{subfigure}[t]{0.8\textwidth}
		\includegraphics[width=1\textwidth]{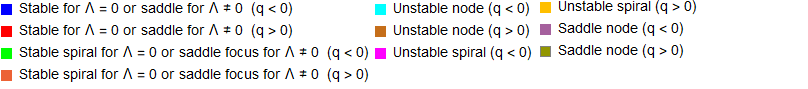}
	\end{subfigure}
	\caption{The bifurcation diagrams depicting the existence and stability conditions of CP4 are plotted with $\mathcal{A}_{\pm}$ defined in equation~\eqref{Aplus} as functions of $\lambda_V$ for the following cases: (a) $d=4, \gamma=0.18$, (b) $d=5, \gamma=0.24$, (c) $d=6, \gamma=0.27$, (d) $d=4, \gamma=0.21$, (e) $d=5, \gamma=0.27$, (f) $d=6, \gamma=0.3$, (g) $d=4, \gamma=0.24$, (h) $d=5, \gamma=0.29$, and $d=6, \gamma=0.32$. The inflation era is represented by the unstable critical points. In the region of stable points, the intersection between the black-thin and black-dashed curves indicates the acceleration in the late-time era comprising scalar fields along with baryonic and dark matter, where $\Omega_\text{m} \approx 0.3$ and $\Omega_\text{KGR} \approx 0.7$. The red curve denotes the parameter $w_\text{KGR} \approx -1$, while the blue curve in Subfig. \ref{subfig:g} (\ref{subfig:h} and \ref{subfig:i}) represents $w_\text{m} = 1$ ($w_\text{m} = 0.45$ and $ w_\text{m} = 0.28$).} 
	\label{fig:cp4t}    
\end{figure*}  
    \item CP5 $\Bigg(0,\sqrt{1-\frac{2\gamma d}{(d-2)}}\Bigg)$:\\ 
    CP5 exists if $\lambda_V=0$ and  $\gamma < \frac{(d-2)}{2d}$. The properties of CP5 can be described by 
    \begin{eqnarray}
        \Omega_\text{KGR} = 1, \quad q=-1\, .
    \end{eqnarray}
    The equation of state parameter for CP5 is
       $ w_\text{KGR} = w_\text{eff}= - 1$.
    The stability properties of CP5 can be seen in Fig. \ref{fig:cp5}. CP5 represents a universe dominated by the scalar field potential. The stabilities are stable and saddle nodes. At this point, the scalar potential term is constant when $\lambda_V = 0$ and behaves like a cosmological constant that dominates the acceleration of the expansion of the universe.
\end{itemize}
%
\section{Dynamical System Analysis: $\Lambda \neq 0$ Case}
\label{sec:Lambdaneq}

In this section, we focus on a model with $\Lambda \neq 0$, which implies the addition of a new dynamical variable related to $\Lambda$. This variable's characteristics significantly influence the history of cosmic evolution, distinguishing it from the case where $\Lambda=0$. However, in the limit of $\lambda_V\rightarrow 0$, making the scalar potential constant ($V_0$), we may observe similar properties as in the $\Lambda \neq 0$ case. Therefore, the analysis discussed in this section can be applied to this latter scenario.

Let us first define a new autonomous variable beside of \eqref{eq:autoxm}, \eqref{eq:autoxp} and \eqref{eq:autoxv} which is associated to $\tilde{\Lambda}$ term,
\begin{eqnarray}
    \label{eq:autoxA}
    x_{\Lambda(\pm)} = \sqrt{\frac{2\kappa\vert\tilde{\Lambda}_\pm\vert}{(d-1)(d-2) H^2}}\, .
\end{eqnarray}
It is important to mention that we consider two cases which are $\tilde{\Lambda}_+ > 0$ and $\tilde{\Lambda}_- < 0$. Positive and negative subscripts correspond to $\tilde{\Lambda} > 0$ and $\tilde{\Lambda} < 0$ cases, respectively. Hence, the equations of motion become,
\begin{align} \label{eq:xphia}
   \frac{2}{(d-1)N} x_{\phi}^{'} =& \left(\frac{\Big[1 - \frac{2\gamma d}{(d-2)}\Big](1+w_\text{m})}{\Big[1-\frac{2\gamma(d-1)}{(d-2)}(1+w_\text{m})\Big]}-2\right)x_{\phi} + \lambda_V \sqrt{\frac{(d-2)}{(d-1)}} x_V^2 \nonumber \\
   &+\frac{(1-w_\text{m})}{\Big[1-\frac{2\gamma(d-1)}{(d-2)}(1+w_\text{m})\Big]} x_{\phi}^3 - \frac{(1+w_\text{m})}{\Big[1-\frac{2\gamma(d-1)}{d-2}(1+w_\text{m})\Big]} x_{\phi} x_{V\Lambda(\pm)}^2 \, ,\\ \nonumber \\
   \label{eq:xva}\frac{2}{(d-1)N} x_V^{'} =& \frac{\Big[1-\frac{2\gamma d}{(d-2)}\Big] (1+w_\text{m})}{\Big[1-\frac{2\gamma(d-1)}{(d-2)}(1+w_\text{m})\Big]} x_V  +\frac{(1 - w_\text{m})}{\Big[1 - \frac{2\gamma(d-1)}{(d-2)}(1 + w_\text{m})\Big]} x_V x_{\phi}^2 \nonumber\\
    &-\lambda_V \sqrt{\frac{(d-2)}{(d-1)}} x_{\phi} x_V - \frac{(1 + w_\text{m})}{\Big[1 -\frac{2\gamma(d-1)}{d-2}(1 + w_\text{m})\Big]} x_V x_{V\Lambda(\pm)}^2 \, ,\\ \nonumber \\
    \label{eq:xLambdaa}\frac{2}{(d-1)N} x_{\Lambda(\pm)}^{'} =& \frac{\left[1 - \frac{2\gamma d}{(d-2)}\right](1 + w_\text{m})}{\Big[1 - \frac{2\gamma(d-1)}{(d-2)}(1 + w_\text{m})\Big]} x_{\Lambda(\pm)} + \frac{(1 - w_\text{m})}{\Big[1 - \frac{2\gamma(d-1)}{(d-2)}(1+w_m)\Big]} x_{\Lambda(\pm)} x_{\phi}^2 \nonumber \\
    &-\frac{(1+w_\text{m})}{\Big[1 - \frac{2\gamma(d-1)}{d-2}(1 + w_\text{m})\Big]}x_{\Lambda(\pm)}x_{V\Lambda(\pm)}^2 \, .
\end{align}
Note that we have introduced the variable $x_{V\Lambda(\pm)}^2$, which satisfies the following conditions,
\begin{eqnarray}
x_{V\Lambda(\pm)}^2 = \begin{cases}
x_V^2 + x_{\Lambda(+)} ^2, & \text{for}\ \tilde{\Lambda}_+ \\
x_V^2 - x_{\Lambda(-)}^2, & \text{for}\ \tilde{\Lambda}_- 
\end{cases} \, .
\end{eqnarray}
Then, the constraint equation becomes,
\begin{equation}
\label{constraineq1}
    \left[1-\frac{2\gamma(d-1)}{(d-2)}(1+w_m)\right] x_m^2 + \left[1 - \frac{4\gamma(d-1)}{(d-2)}\right] x_{\phi}^2  + x_{V\Lambda(\pm)}^2 = 1 - \frac{2\gamma d}{(d-2)}\, ,
\end{equation}
which can be written in terms of the density parameters as,
\begin{equation}
     \Omega_m + \Omega_\text{KGR} \pm \Omega_{\Lambda\, (\pm)} =  1 \, ,
\end{equation}
where we have defined a density parameter of the scalar field $\Omega_\text{KGR}$ as in equation \eqref{eq:Omegaphi}. Here, we have a new density parameter $\Omega_{\Lambda\, (\pm)}$ defined as 
\begin{align}
\Omega_{\Lambda \, (\pm)} \equiv & \frac{2\kappa\vert\tilde{\Lambda}_\pm\vert}{(d-1)(d-2) H^2}= x_{\Lambda(\pm)}^2 \, .
\end{align}
The deceleration parameter becomes
\begin{equation}
    q = -1 + \frac{(d-1)}{2 \Big[1-\frac{2\gamma(d-1)}{(d-2)}\Big]} \left[1 - \frac{2\gamma d}{(d-2)}+w_\text{m} x_m^2 + x_{\phi}^2 - x_{V\Lambda(\pm)}^2\right]\, ,
\end{equation}
where $x_m^2$ can be obtained from the constraint equation \eqref{constraineq1}. In this model, we have five critical points. We denote each critical point ($x_{\phi,c}$, $x_{V,c}$, $x_{\Lambda_{,c(\pm)}}$) as CP$i$ where $i=1,2,3,4,5$. The variables $x_{\phi,c}$, $x_{V,c}$ and $x_{\Lambda_{,c(\pm)}}$ represent the critical values of $x_{\phi}$, $x_{V}$ and $x_{\Lambda_{(\pm)}}$, respectively, which trigger the right-hand side of Eq. \eqref{eq:xphia}-\eqref{eq:xLambdaa} to become zero. The critical points for $\Lambda \neq 0$ case are given by the following lists: 

\begin{itemize}
    \item CP1 $(0,0,0)$\\
      The equation of state parameter for the scalar field $w_\text{KGR}$ is undefined if $ \gamma = 0$. Otherwise, $w_\text{KGR} = -1$. The existence and the stability properties of CP1 can be seen in Fig. \ref{fig:cp1}. In this case, $\Lambda\neq 0$, it describes the matter-dominated era in which it has the same properties and behavior as  CP1 in the absence of $\Lambda$. CP1 has the density parameter $\Omega_\text{KGR}$ and the deceleration $q$ given by \eqref{eq:qcp1}.
    \item CP2 $\left(\sqrt{1-\frac{2\gamma(1+w_\text{m})}{(1-w_\text{m})}},0,0\right)$ and CP3 $\left(-\sqrt{1-\frac{2\gamma(1+w_\text{m})}{(1-w_\text{m})}},0,0\right)$:\\ \\
    CP2 and CP3 exist if $w_m\neq 1$ and $1-\frac{2\gamma(1+w_\text{m})}{(1 - w_\text{m})}>0$.
    The equation of state parameter is given by equation \eqref{eq:eomcp2}. The stability properties of CP2 and CP3 can be seen in Fig. \ref{fig:cp2} and Fig. \ref{fig:cp3}, respectively. For the $\Lambda\neq 0$ case, CP2 and CP3 describe the kinetic-matter dominated era in which it has no stable nodes or late-time attractors as obtained in the $\Lambda = 0$ case. Two possible stabilities, namely the unstable nodes and the saddle nodes, can be obtained at these points depending on the parameter values. CP2 and CP3 have the density parameter $\Omega_\text{KGR}$ and the deceleration parameter $q$ given by  \eqref{eq:qcp2}.
    \item CP4 $\Bigg(\frac{1}{\lambda_V} \sqrt{\frac{d-1}{d-2}} \mathcal{A}_{\pm} , \frac{1}{\lambda_V} \sqrt{\frac{(d-1)(2 - \mathcal{A}_{\pm}) 
    		\mathcal{A}_{\pm}}{(d-2)}}, 0\Bigg)$:\\ \\
    CP4 exists if  $0 < \mathcal{A}_{\pm} < 2$. The equation of state parameter is given by equation \eqref{eq:eomcp4}. The parameter $\mathcal{A}_{\pm}$ is defined as equation \eqref{Aplus}. This point has the density parameters $\Omega_\text{KGR}$ and $q$ given by equation \eqref{eq:Ocp4} and \eqref{eq:qcp4}. The stability properties of CP4 are shown in Fig. \ref{fig:cp4t}. The stability properties and behavior of CP4 found for $\Lambda\neq 0$ are distinct from the $\Lambda=0$ case in which we do not obtain late-time attractors. It implies that CP4 cannot be a late-time model of the universe due to the none existence of stable nodes. However, this point remains to describe the early universe known as power-law inflation.
    \item CP5 $\left(0,\sqrt{1-\frac{2\gamma d}{(d-2)}\mp x_{\Lambda,c(\pm)}^2},x_{\Lambda,c \, (\pm)}\right)$:\\ \\
    CP5 exists for $\lambda_V=0$, $0\leq x_{\Lambda,c\,(+)}\leq \sqrt{1-\frac{2\gamma d}{(d-2)}}$ for $\tilde{\Lambda} > 0$ and $ x_{\Lambda,\,c(-)}\geq 0$ for $\tilde{\Lambda} < 0$. The properties of CP5 can be described by the following parameters,
    \begin{eqnarray}
   \Omega_\text{KGR} \pm  \Omega_{\Lambda\, (\pm)} = 1, \quad q=-1 \, .
   \end{eqnarray}
   The stability properties of CP5 can be seen in Fig. \ref{fig:cp5} and Fig. \ref{fig:cp5a}. CP5 depicts a cosmological constant-like dominated accelerating universe, which is contributed by a constant potential term of the scalar field, $\lambda_V = 0$, and $\Lambda$. The stability is either incomplete stable or unstable. The incomplete stable indicate the late-time universe where the interplay between massless scalar field and cosmological constant acts as the dark energy. The unstable may indicate the early-time attractors known as conventional inflation. The effective equation of state parameter in this critical point is
   \begin{equation}
       w_{\text{eff}} \equiv \frac{p_\text{KGR}+p_{\Lambda}}{\rho_\text{KGR}+\rho_{\Lambda}} = -1 \, .
   \end{equation}
\begin{figure}[t]
	\centering
	\begin{subfigure}[t]{0.32\textwidth}
		\includegraphics[width=1\textwidth]{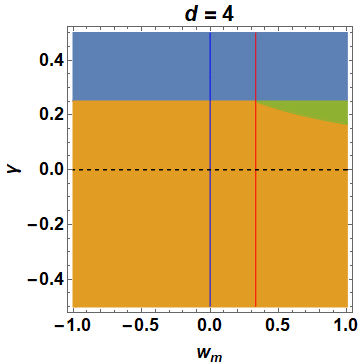}
		\caption{$ d=4$}
	\end{subfigure}   
	\hfill
	\begin{subfigure}[t]{0.32\textwidth}
		\includegraphics[width=1\textwidth]{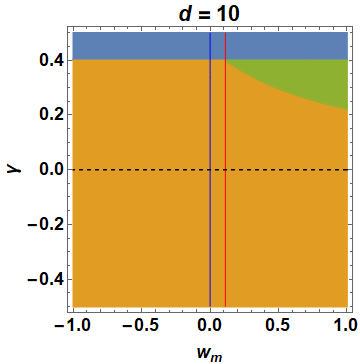}
		\caption{$d=10$}
	\end{subfigure} 
	\hfill
	\begin{subfigure}[t]{0.32\textwidth}
		\includegraphics[width=1\textwidth]{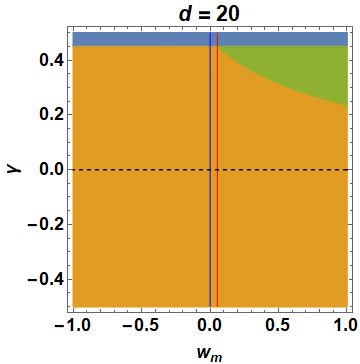}
		\caption{$d=20$}
	\end{subfigure}   
	\vfill
	\begin{subfigure}[t]{0.5\textwidth}
		\includegraphics[width=1\textwidth]{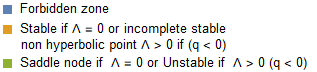}
	\end{subfigure}              
	\caption{The bifurcation diagrams show the existence and the stability conditions of CP5. The stability properties of CP5 by plotting $\gamma$ as a function of $w_\text{m}$ for constant value of (a) $d=4$, (b) $d=10$ and (c) $d=20$. The stability characteristics of this point are given in both $\Lambda=0$ and $\Lambda > 0$ cases which can be seen in the figure legends. The blue line and the red line represent $w_\text{m} = 0$ and $w_\text{m} = \frac{1}{(d-1)}$, respectively. Note that we have a forbidden zone giving $\sqrt{1-\frac{2\gamma d}{(d-2)}} < 0 $ for $\Lambda=0$ or $\sqrt{1-\frac{2\gamma d}{(d-2)}-x_{\Lambda,c\,(+)}^2} < 0 $ for $\Lambda > 0$ case.}
	\label{fig:cp5}    
\end{figure}
\begin{figure}[t]
	\centering
	\begin{subfigure}[t]{0.32\textwidth}
		\includegraphics[width=1\textwidth]{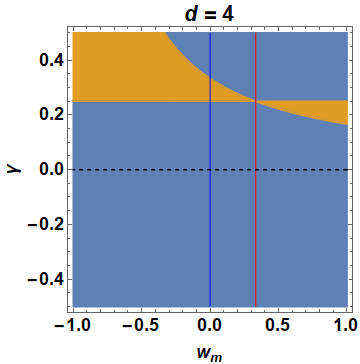}
		\caption{$ d=4$}
	\end{subfigure}   
	\hfill
	\begin{subfigure}[t]{0.32\textwidth}
		\includegraphics[width=1\textwidth]{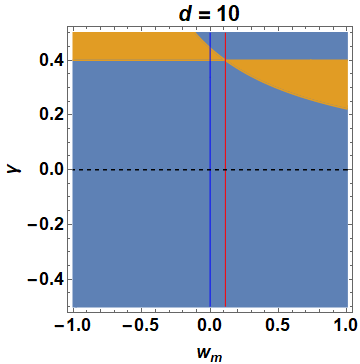}
		\caption{$d=10$}
	\end{subfigure} 
	\hfill
	\begin{subfigure}[t]{0.32\textwidth}
		\includegraphics[width=1\textwidth]{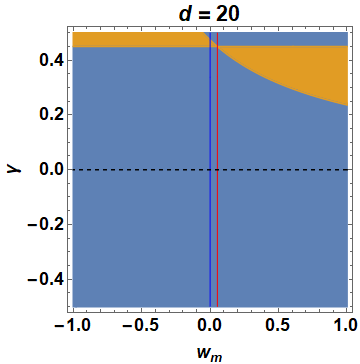}
		\caption{$d=20$}
	\end{subfigure}   
	\vfill
	\begin{subfigure}[t]{0.45\textwidth}
		\includegraphics[width=1\textwidth]{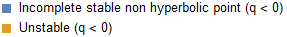}
	\end{subfigure}              
	\caption{The bifurcation diagrams show the existence and the stability conditions of CP5. The stability properties of CP5 by plotting $\gamma$ as a function of $w_\text{m}$ for constant value of (a) $d=4$, (b) $d=10$ and (c) $d=20$. The stability characteristics of this point are given in both $\Lambda=0$ and $\Lambda > 0$ cases which can be seen in the figure legends. The blue line and the red line represent $w_\text{m} = 0$ and $w_\text{m} = \frac{1}{(d-1)}$, respectively. Note that we have a forbidden zone giving $\sqrt{1-\frac{2\gamma d}{(d-2)}} < 0 $ for $\Lambda=0$ or $\sqrt{1-\frac{2\gamma d}{(d-2)}-x_{\Lambda,c\,(+)}^2} < 0 $ for $\Lambda > 0$ case.}
	\label{fig:cp5a}    
\end{figure}
\end{itemize}
\section{Local-Global Existence of Solutions}
\label{sec:LocGlobExis}

In this Section, we establish the local-global existence and the uniqueness of the evolution equations in the preceding sections using Picard's iteration and the contraction mapping properties. We begin the construction by considering the $\Lambda = 0$ case, and then, the $\Lambda \neq 0$ case.  

First of all, let us consider the $\Lambda = 0$ case in which we define the dynamical variables
\begin{equation}
    \bm{u} = \begin{pmatrix} x_\phi \\ x_V  \end{pmatrix}, \label{dynvar}
\end{equation}
 on an interval $I \equiv [s, s+ \epsilon]$ where $s \equiv \ln{a} \in \lR$ and $\varepsilon$ is a small positive constant. If all constants in \eqref{constraineq} are positive, that is, it satisfies
\begin{equation}\label{eq:gammacond}
\gamma \le \frac{d-2}{4(d-1)} ~ , \quad \gamma(1+w_\text{m}) \le \frac{d-2}{2(d-1)} ~ ,
\end{equation}
 then, we could have
\begin{equation}
    \left.
    \begin{array}{cc}
        0 \leq  \left|1-\frac{4\gamma(d-1)}{(d-2)}\right|^{\frac{1}{2}} | x_\phi | \leq  \left| 1-\frac{2\gamma d}{(d-2)} \right|^{\frac{1}{2}} & ~ , \\
        0 \leq | x_V | \leq \left| 1-\frac{2\gamma d}{(d-2)} \right|^{\frac{1}{2}} &  ~ ,\\
            \end{array}
    \right. \label{eq:syaratvardin}
\end{equation}
implying that  all of the quantities $\left( \left|1-\frac{4\gamma(d-1)}{(d-2)}\right|^{\frac{1}{2}}  x_\phi, x_V \right)$ are defined on an open set $U \subset S^2$ where $S^2$ is the 2-sphere with radius $ \left| 1-\frac{2\gamma d}{(d-2)} \right|^{\frac{1}{2}}$.

Without loss of generality, we could set the lapse function $N =1$ so that  the evolution equations \eqref{eq:xphi} and  \eqref{eq:xv} can be simply rewritten into
\begin{equation}
    \frac{d\bm{u}}{ds} = \mathcal{J}(\bm{u}) ~ , \label{fungsiJ}
\end{equation}
with
\begin{eqnarray}\label{JY}
    \mathcal{J}(\bm{u}) \equiv \frac{1}{2} (d-1) \left( \begin{array}{c}
        \left(\frac{\Big[1 - \frac{2\gamma d}{(d-2)}\Big](1+w_\text{m})}{\Big[1-\frac{2\gamma(d-1)}{(d-2)}(1+w_\text{m})\Big]}-2\right)x_{\phi} + \lambda_V \sqrt{\frac{d-2}{d-1}} x_V^2  \\
   + \frac{(1-w_\text{m})}{\Big[1-\frac{2\gamma(d-1)}{(d-2)}(1+w_\text{m})\Big]} x_{\phi}^3 - \frac{(1+w_\text{m})}{\Big[1-\frac{2\gamma(d-1)}{(d-2)}(1+w_\text{m})\Big]} x_{\phi} x_V^2 \\ \\
        \frac{\Big[1-\frac{2\gamma d}{(d-2)}\Big] (1+w_\text{m})}{\Big[1-\frac{2\gamma(d-1)}{(d-2)}(1+w_\text{m})\Big]} x_V - \lambda_V \sqrt{\frac{d-2}{d-1}} x_{\phi} x_V \\
    +\frac{(1-w_\text{m})}{\Big[1-\frac{2\gamma(d-1)}{(d-2)}(1+w_\text{m})\Big]} x_V x_{\phi}^2 - \frac{(1+w_\text{m})}{\Big[1-\frac{2\gamma(d-1)}{(d-2)}(1+w_\text{m})\Big]} x_V^3 \\ \\
\end{array} \right).
\end{eqnarray}

\begin{lemma} \label{opJY}
    The operator $ \mathcal{J}(\bm{u})$ in Eq.~\eqref{fungsiJ} is locally Lipschitz with respect to $\bm{u}$.
\end{lemma}

\begin{proof}
We have the following estimate
\begin{eqnarray}
    && | \mathcal{J} |_U \leq \frac{1}{2} (d-1) \left[ \left| \frac{\Big[1 - \frac{2\gamma d}{(d-2)}\Big](1+w_\text{m})}{\Big[1-\frac{2\gamma(d-1)}{(d-2)}(1+w_\text{m})\Big]}-2\right| |x_{\phi}| + |\lambda_V| \sqrt{\frac{d-2}{d-1}} |x_V|^2\right. \nonumber \\
    && \qquad + \, \left| \frac{(1-w_\text{m})}{\Big[1-\frac{2\gamma(d-1)}{(d-2)}(1+w_\text{m})\Big]}\right| |x_{\phi}|^3  + \left| \frac{(1+w_\text{m})}{\Big[1-\frac{2\gamma(d-1)}{(d-2)}(1+w_\text{m})\Big]}\right|  |x_{\phi}| |x_V|^2 \nonumber \\
    && \qquad +  \left| \frac{\Big[1-\frac{2\gamma d}{(d-2)}\Big] (1+w_\text{m})}{\Big[1-\frac{2\gamma(d-1)}{(d-2)}(1+w_\text{m})\Big]}\right| |x_V| +  |\lambda_V | \sqrt{\frac{d-2}{d-1}} |x_{\phi}| |x_V| \nonumber \\
    && \qquad \left. + \, \left| \frac{(1-w_\text{m})}{\Big[1-\frac{2\gamma(d-1)}{(d-2)}(1+w_\text{m})\Big]} \right| |x_V| |x_{\phi}|^2  + \left| \frac{(1+w_\text{m})}{\Big[1-\frac{2\gamma(d-1)}{(d-2)}(1+w_\text{m})\Big]} \right| |x_V|^3 \right]. \label{estJY}
\end{eqnarray}
Then, using Eq.~\eqref{eq:syaratvardin}, we can show that $| \mathcal{J}(\bm{u}) |_U$ is indeed bounded on $U$.

Moreover, for $\bm{u}, \hat{\bm{u}} \in U$ we have
\begin{eqnarray}\label{estJYLps}
  \hspace{-0.2cm}  | \mathcal{J}(\bm{u}) - \mathcal{J}(\hat{\bm{u}}) |_U &\leq& \frac{1}{2} (d-1) \left[ \left| \frac{\Big[1 - \frac{2\gamma d}{(d-2)}\Big](1+w_\text{m})}{\Big[1-\frac{2\gamma(d-1)}{(d-2)}(1+w_\text{m})\Big]}-2\right| |x_{\phi} - \hat{x}_{\phi}| \right. \nonumber \\
    && \!\!\!\!\! + |\lambda_V| \sqrt{\frac{d-2}{d-1}} \big(|x^2_V -\hat{x}^2_V| + |x_{\phi} x_V - \hat{x}_{\phi} \hat{x}_V|\big)  \nonumber \\
    && \!\!\!\!\! +  \left| \frac{(1-w_\text{m})}{\Big[1-\frac{2\gamma(d-1)}{(d-2)}(1+w_\text{m})\Big]} \right| \big(|x^3_{\phi} - \hat{x}^3_{\phi}| + |x_V x^2_{\phi} - \hat{x}_V \hat{x}^2_{\phi}|\big) \nonumber \\
    &&\!\!\!\!\! + \left| \frac{(1+w_\text{m})}{\Big[1-\frac{2\gamma(d-1)}{(d-2)}(1+w_\text{m})\Big]} \right| \big(|x_{\phi} x^2_V - \hat{x}_{\phi} \hat{x}^2_V| + |x^3_V - \hat{x}^3_V|\big)  \nonumber \\
    &&\!\!\!\!\! \left. +  \left| \frac{\Big[1-\frac{2\gamma d}{(d-2)}\Big] (1+w_\text{m})}{\Big[1-\frac{2\gamma(d-1)}{(d-2)}(1+w_\text{m})\Big]}\right| |x_V -\hat{x}_V| \right] . 
\end{eqnarray}
After some computations, we could show that $\mathcal{J}$ is locally Lipschitz with respect to $\bm{u}$, namely,
\begin{equation}
    \left| \mathcal{J}(\bm{u}) - \mathcal{J}(\hat{\bm{u}}) \right|_U \le C_{\mathcal{J}}(| \bm{u} |, | \hat{\bm{u}} |) | \bm{u} - \hat{\bm{u}} | ~ , \label{localLipshitzcon}
\end{equation}
where $C_{\mathcal{J}} (| \bm{u} |, | \hat{\bm{u}} |)$ is a local Lipschitz constant.
\end{proof}

Next, the integral form of Eq.~\eqref{fungsiJ} is given by
\begin{equation}
    \bm{u}(s) = \bm{u}(s_0) + \int_{s_0}^s \, \mathcal{J} \left( \bm{u}(\hat{s}) \right) d\hat{s} ~ . \label{IntegralEquation}
\end{equation}
By defining a Banach space
\begin{equation}
    X \equiv \{ \bm{u} \in C(I, \mathbb{R}^2) : \, \bm{u}(x_0) = \bm{u}_{0}, \, \sup_{x \in I}{| \bm{u}(x) |} \leq L_0 \} ~ ,
\end{equation}
endowed with the norm
\begin{equation}
    | \bm{u} |_{X} = \sup_{x \in I}{|\bm{u}(x)|} ~ ,
\end{equation}
where $L_0 > 0$, we introduce an operator $\mathcal{K}$
\begin{equation}
    \mathcal{K}(\bm{u}(x)) = \bm{u}_{0} + \int_{x_0}^x \mathcal{J} \left( \bm{u}(s), s \right) ds ~ , \label{OpKdefinition}
\end{equation}
and using Lemma \ref{opJY}, we have the following result \cite{akbar2015existence}:
\begin{lemma} \label{uniqueness}
    Let $\mathcal{K}$ be an operator defined in Eq.~\eqref{OpKdefinition}. Suppose there exists a constant $\varepsilon > 0$ such that $\mathcal{K}$ is a mapping from $X$ to itself and $\mathcal{K}$ is a contraction mapping on $I = [x, x + \varepsilon]$ with
    \begin{equation}
        \varepsilon \leq \min \left( \frac{1}{C_{L_0}}, \frac{1}{C_{L_0} L_0 + \| \mathcal{J}(x) \|} \right) ~ .
    \end{equation}
    Then, the operator $\mathcal{K}$ is a contraction mapping on $X$.
\end{lemma}
\noindent which shows the existence of a unique fixed point of Eq.~\eqref{OpKdefinition} ensuring a unique local solution of the differential equation \eqref{fungsiJ}. One can further establish a maximal solution by repeating the above arguments of the local existence with the initial condition $\bm{u}(x - x_n)$ for some $x_0 < x_n < x$ and using the uniqueness condition to glue the solutions.

To show the global existence of  Eq.~\eqref{fungsiJ}, let us first consider an inequality coming from \eqref{IntegralEquation} 
\begin{equation}
    | \bm{u}(s) | \le | \bm{u}(s_0) | + \int_{s_0}^s | \mathcal{J} \left( \bm{u}(\hat{s}) \right) | d\hat{s}. \label{IntegralEquation1}
\end{equation}
 Using Eqs.~\eqref{eq:syaratvardin} and \eqref{estJY}, we get
 %
\begin{eqnarray}
    | \bm{u}(t) | &\leq& | \bm{u}(t_0) | + \frac{1}{2} (d-1) \left[ \left| \frac{\Big[1 - \frac{2\gamma d}{(d-2)}\Big](1+w_\text{m})}{\Big[1-\frac{2\gamma(d-1)}{(d-2)}(1+w_\text{m})\Big]}-2\right|  \frac{\left| 1-\frac{2\gamma d}{(d-2)} \right|^{\frac{1}{2}}}{\left|1-\frac{4\gamma(d-1)}{(d-2)}\right|^{\frac{1}{2}}} \right. \nonumber \\
    && + |\lambda_V| \sqrt{\frac{d-2}{d-1}} \left(\left| 1-\frac{2\gamma d}{(d-2)} \right| + \frac{\left| 1-\frac{2\gamma d}{(d-2)} \right|}{\left|1-\frac{4\gamma(d-1)}{(d-2)}\right|^{\frac{1}{2}}} \right)  \nonumber \\
    &&  + \left| \frac{(1-w_\text{m})}{\Big[1-\frac{2\gamma(d-1)}{(d-2)}(1+w_\text{m})\Big]} \right|  \left(\frac{\left| 1-\frac{2\gamma d}{(d-2)} \right|^{\frac{3}{2}}}{\left|1-\frac{4\gamma(d-1)}{(d-2)}\right|} + \frac{\left| 1-\frac{2\gamma d}{(d-2)} \right|^{\frac{3}{2}}}{\left|1-\frac{4\gamma(d-1)}{(d-2)}\right|^{\frac{3}{2}}}\right)\nonumber \\
    &&  + \left| \frac{(1+w_\text{m})}{\Big[1-\frac{2\gamma(d-1)}{(d-2)}(1+w_\text{m})\Big]} \right| \left(\frac{\left| 1-\frac{2\gamma d}{(d-2)} \right|^{\frac{3}{2}}}{\left|1-\frac{4\gamma(d-1)}{(d-2)}\right|^{\frac{1}{2}}} + \left| 1-\frac{2\gamma d}{(d-2)} \right|^{\frac{3}{2}}\right)  \nonumber\\
    &&  \left. +  \left| \frac{\Big[1-\frac{2\gamma d}{(d-2)}\Big] (1+w_\text{m})}{\Big[1-\frac{2\gamma(d-1)}{(d-2)}(1+w_\text{m})\Big]}\right| \left| 1-\frac{2\gamma d}{(d-2)} \right|^{\frac{1}{2}}\right] \ln{\left( \frac{a(t)}{a(t_0)} \right)} ~ .
    \label{eq:solIntegralEquation}
\end{eqnarray}
%
The second part is to consider a case where
\begin{equation}\label{eq:gammacond1}
\gamma > \frac{d-2}{4(d-1)} ~ , \quad \gamma(1+w_\text{m}) \le \frac{d-2}{2(d-1)} ~ ,
\end{equation}
in which the pre-coefficient of the second term in the left-hand side of \eqref{constraineq} is negative. In this case, we may have
\begin{equation}
    \left.
    \begin{array}{rcl}
       \left|1-\frac{4\gamma(d-1)}{(d-2)}\right|^{\frac{1}{2}}   x_\phi & = & \left| 1-\frac{2\gamma d}{(d-2)} \right|^{\frac{1}{2}} \cos{\alpha} \sinh{\beta} \\
        x_V & = & \left| 1-\frac{2\gamma d}{(d-2)} \right|^{\frac{1}{2}} \cos{\alpha} \cosh{\beta} \\
         \left| 1-\frac{2\gamma(d-1)}{(d-2)}(1+w_m)\right|^{\frac{1}{2}}  x_m & = & \left| 1-\frac{2\gamma d}{(d-2)} \right|^{\frac{1}{2}} \sin{\alpha} 
    \end{array}
    \right. \label{eq:syaratvardin1}
\end{equation}
where $\alpha \equiv \alpha(s)$ and $\beta \equiv \beta(s)$. In the case at hand, it is easy to show that Lemma \ref{opJY} and Lemma \ref{uniqueness} still hold, but the estimate of \eqref{IntegralEquation1} has to be modified. Thus, the estimate \eqref{IntegralEquation1} has the form
\begin{eqnarray}
 \!\!\!\! \!\!\!\! \!\!\!\!   | \bm{u}(t) | &\leq& | \bm{u}(t_0) | + \frac{1}{2} (d-1) \left[ \left| \frac{\Big[1 - \frac{2\gamma d}{(d-2)}\Big](1+w_\text{m})}{\Big[1-\frac{2\gamma(d-1)}{(d-2)}(1+w_\text{m})\Big]}-2\right| \frac{\left| 1-\frac{2\gamma d}{(d-2)} \right|^{\frac{1}{2}}}{\left|1-\frac{4\gamma(d-1)}{(d-2)}\right|^{\frac{1}{2}}} \int_{s_0}^s |\sinh{\beta}| d\hat{s} \right. \nonumber \\
    && + |\lambda_V| \sqrt{\frac{d-2}{d-1}} \left|1-\frac{4\gamma(d-1)}{(d-2)}\right| \int_{s_0}^s \cosh^2{\beta} ~ d\hat{s} \nonumber \\
    && +  \left| \frac{(1-w_\text{m})}{\Big[1-\frac{2\gamma(d-1)}{(d-2)}(1+w_\text{m})\Big]}\right| \frac{\left| 1-\frac{2\gamma d}{(d-2)} \right|^{\frac{3}{2}}}{\left|1-\frac{4\gamma(d-1)}{(d-2)}\right|^{\frac{3}{2}}} \int_{s_0}^s |\sinh{\beta}|^3 d\hat{s}  \nonumber \\
    &&   + \left| \frac{(1+w_\text{m})}{\Big[1-\frac{2\gamma(d-1)}{(d-2)}(1+w_\text{m})\Big]}\right| \frac{\left| 1-\frac{2\gamma d}{(d-2)} \right|^{\frac{3}{2}}}{\left|1-\frac{4\gamma(d-1)}{(d-2)}\right|^{\frac{1}{2}}} \int_{s_0}^s |\sinh{\beta}|\cosh^2{\beta}  d\hat{s} \nonumber \\
    &&  +  \left| \frac{\Big[1-\frac{2\gamma d}{(d-2)}\Big] (1+w_\text{m})}{\Big[1-\frac{2\gamma(d-1)}{(d-2)}(1+w_\text{m})\Big]}\right| \left|1-\frac{4\gamma(d-1)}{(d-2)}\right|^{\frac{1}{2}} \int_{s_0}^s |\cosh{\beta}| ~ d\hat{s} \nonumber \\
    && +  |\lambda_V | \sqrt{\frac{d-2}{d-1}}  \frac{\left| 1-\frac{2\gamma d}{(d-2)} \right|}{\left|1-\frac{4\gamma(d-1)}{(d-2)}\right|^{\frac{1}{2}}} \int_{s_0}^s |\sinh{\beta} \cosh{\beta}| d\hat{s}  \nonumber \\
    &&   +  \left| \frac{(1-w_\text{m})}{\Big[1-\frac{2\gamma(d-1)}{(d-2)}(1+w_\text{m})\Big]} \right| \frac{\left| 1-\frac{2\gamma d}{(d-2)} \right|^{\frac{3}{2}}}{\left|1-\frac{4\gamma(d-1)}{(d-2)}\right|} \int_{s_0}^s |\sinh^2{\beta} \cosh{\beta}| d\hat{s}  \nonumber\\
    && \left. + \left| \frac{(1+w_\text{m})}{\Big[1-\frac{2\gamma(d-1)}{(d-2)}(1+w_\text{m})\Big]} \right| \left|1-\frac{4\gamma(d-1)}{(d-2)}\right|^{\frac{3}{2}} \int_{s_0}^s |\cosh^3{\beta}| ~ d\hat{s}  \right]. \label{eq:solIntegralEquation1}
\end{eqnarray}
 For other cases, we employ a similar methodology as above . 

Finally, we discuss the $\Lambda \neq 0$ case in which we introduce the extended dynamical variables
\begin{equation}
    \bm{u}_\Lambda  = \begin{pmatrix} x_\phi \\ x_V \\ x_\Lambda \end{pmatrix}, \label{dynvarlamb}
\end{equation}
 on an interval $I \equiv [s, s+ \epsilon]$ such that we have the equation
\begin{equation}
    \frac{d\bm{u}_\Lambda}{ds} = \mathcal{J}_\Lambda(\bm{u}_\Lambda) ~ , \label{fungsiJlamb}
\end{equation}
coming from \eqref{eq:xphia}, \eqref{eq:xva}, and \eqref{eq:xLambdaa} where
%
\begin{eqnarray}\label{JYlamb}
    \mathcal{J}_\Lambda(\bm{u}_\Lambda) \equiv \frac{1}{2} (d-1) \left( \begin{array}{c}
        \left(\frac{\Big[1 - \frac{2\gamma d}{(d-2)}\Big](1+w_\text{m})}{\Big[1-\frac{2\gamma(d-1)}{(d-2)}(1+w_\text{m})\Big]}-2\right)x_{\phi} + \lambda_V \sqrt{\frac{d-2}{d-1}} x_V^2 \nonumber \\
   +\frac{(1-w_\text{m})}{\Big[1-\frac{2\gamma(d-1)}{(d-2)}(1+w_\text{m})\Big]} x_{\phi}^3 - \frac{(1+w_\text{m})}{\Big[1-\frac{2\gamma(d-1)}{(d-2)}(1+w_\text{m})\Big]} x_{\phi} x_{V\Lambda(\pm)}^2 \\ \\
        \frac{\Big[1-\frac{2\gamma d}{(d-2)}\Big] (1+w_\text{m})}{\Big[1-\frac{2\gamma(d-1)}{(d-2)}(1+w_\text{m})\Big]} x_V  +\frac{(1 - w_\text{m})}{\Big[1 - \frac{2\gamma(d-1)}{(d-2)}(1 + w_\text{m})\Big]} x_V x_{\phi}^2 \nonumber\\
    -\lambda_V \sqrt{\frac{d-2}{d-1}} x_{\phi} x_V - \frac{(1 + w_\text{m})}{\Big[1 -\frac{2\gamma(d-1)}{(d-2)}(1 + w_\text{m})\Big]} x_V x_{V\Lambda(\pm)}^2  \\ \\
        \frac{\left[1 - \frac{2\gamma d}{(d-2)}\right](1 + w_\text{m})}{\Big[1 - \frac{2\gamma(d-1)}{(d-2)}(1 + w_\text{m})\Big]} x_{\Lambda(\pm)} + \frac{(1 - w_\text{m})}{\Big[1 - \frac{2\gamma(d-1)}{(d-2)}(1+w_m)\Big]} x_{\Lambda(\pm)} x_{\phi}^2 \nonumber \\
    -\frac{(1+w_\text{m})}{\Big[1 - \frac{2\gamma(d-1)}{(d-2)}(1 + w_\text{m})\Big]}x_{\Lambda(\pm)}x_{V\Lambda(\pm)}^2 
\end{array} \right).
\end{eqnarray}
%
Then, we employ a similar procedure as the preceding $\Lambda = 0$ case to show that in this case a global solution of \eqref{eq:xphia}, \eqref{eq:xva}, and \eqref{eq:xLambdaa} with constraint \eqref{constraineq1} does exist.

Thus, we could state
\begin{theorem} \label{thmlocglob}
    There exists a global classical solution of spatially flat FLRW spacetimes in higher dimensional Klein-Gordon-Rastall theory with a scalar potential \eqref{eq:Vpotexp} and a real cosmological constant.
\end{theorem}

\section{Cosmological Model}
\label{sec:cosmologicalmodel}
In Sec. \eqref{sec:dynamicalLambda} and \eqref{sec:Lambdaneq}, we have described the cosmological behaviors of each critical point in our model. In this section, We focus on discussing the cosmological model based on them, specifically on how our model might provide pictures of the early-time universe corresponding to its unstable nodes whilst the late-time universe corresponds to its stable nodes. The former is associated with the inflationary phase which we shall show that even within the KGR framework we still obtain power-law inflation due to the exponential potential of the scalar field. The latter one, the late-time universe model, should be able to explain the accelerated expansion of the universe which is driven by the dark energy-like along with baryonic and dark matter as the constituents of the universe. For the $\Lambda = 0$ case, these interesting features are included in CP4 since it gives accelerated expansion both at early and late-time eras. At this point, CP4, the cosmological model of $\Lambda \neq 0$ case cannot describe the late-time universe because there are only unstable and saddle nodes. The stable nodes, in this case, are only obtained from CP5 even if we have to argue that the non-hyperbolic regions will become late-time attractors by further analysis, as CP5 is reached as $\lambda_V = 0$, this point does not describe the universe filled by the scalar field with exponential potential such that it is inconsistent to power-law inflation in the early time.\\ \\
\subsection{The Early Time Universe: Power-Law Inflation in KGR Cosmology}
To convert Friedmann equation \eqref{eq:friedmann} into autonomous equations, we have introduced the exponential potential of $V(\phi)$ as the specific form in our analysis. Since we use the exponential form of the potential, we may have an inflationary era in the early epoch of the universe that can be described by the power-law inflation if the scalar $\phi$ plays a role as inflaton field which dominates the energy density of the early universe. In this case, the scale factor is given by $a(t) \propto t^{\ell}$ \cite{lucchin1985power}, where the parameter $\ell > 1$. Thus, CP4 with saddle nodes and $\Omega_\text{KGR} \approx 1 $, are good candidates to describe such an era. On the side, some models describe an inflationary era in which not only did the inflaton predominate in the early universe, $\Omega_\text{KGR} < 1 $, but was also perfect fluid with a general value of $w_\text{m}$ \cite{barrow1993scalar, banerjee1998power}. In CP4, the latter possibility gives unstable solutions and an accelerating universe which makes this condition can also be inferred as a power-law inflation model in Rastall cosmology as long as the scale factor evolves according to $a(t) \propto t^{\ell}$ with $\ell > 1$. \\ \\
Now we discuss the inflationary model obtained from CP4 for both $\Lambda = 0$ and $\Lambda \neq 0$ cases. We have to emphasize that Rastall theory, $\gamma \neq 0$, is only consistent with the inflationary phase driven by inflaton and perfect fluid. Let us begin with the Friedmann equation \eqref{eq:friedmann}, by taking $N = 1$, under the condition that not only does the scalar $\phi$ energy density dominates at early epoch but also $\rho_\text{m}$ exists in some fraction as a perfect fluid,

\begin{align}
H ^2 =& \frac{2\kappa}{(d-1) (d - 2 - 2\gamma d)} \Bigg(\bigg[1 - \frac{4 \gamma(d-1)}{(d-2)} \bigg] \frac{\dot{\phi}^2}{2} + V(\phi) + \left[1 - \frac{2\gamma (d-1)}{(d-2)} (1 + w_\text{m})\right] \rho_\text{m} \Bigg) \, . \label{eq:friedmanninf}
\end{align}
In this situation, the perfect fluid may be a form of dust-like $(w_\text{m} = 0)$, radiation-like $\left(w_\text{m} = \frac{1}{(d-1)}\right)$, vacuum-like $(w_\text{m} = -1)$, or stiff-like ($w_\text{m} = 1$). Another possibility is that it is a fluid manifestation of the scalar field other than the inflaton that exists during inflation, as long as they are equivalent \cite{arroja2010note, faraoni2012correspondence,diez2013note}. Equation of motion of scalar field $\phi$ gives,
\begin{equation}
\label{eq:sf}
    \ddot{\phi} + (d-1)H\dot{\phi} + V,_{\phi} = 0\, . 
\end{equation}
We also have
\begin{eqnarray}
\label{eq:C}
    \dot{H} = -\frac{\kappa}{(d-2)} \Big[\rho_\text{m}+p_\text{m} + \dot{\phi}^2\Big] \, ,
\end{eqnarray}
and
\begin{align}
  \ddot{H} = - \frac{\kappa}{(d-2)} \Bigg[\dot{\rho}_\text{m} + \dot{p}_\text{m} + 2\dot{\phi}\ddot{\phi}\Bigg] 
 \, . \label{D}
\end{align}
In addition, we have the solution of Eq. \eqref{eq:sf} as follows,

\begin{eqnarray}
    \phi(t) = \frac{2}{\sqrt{\kappa} \lambda_V} \ln{(\kappa^{-\frac{1}{2}}\, t)}\, ,
\end{eqnarray}
such that we have,
\begin{eqnarray}
    V_0 = \frac{2}{\kappa^2 \lambda_V^2} [(d-1)\ell - 1]\, ,
\end{eqnarray}
From Rastall modification, we have
\begin{eqnarray}
    \dot{\rho}_\text{m} + (d-1)H(\rho_\text{m} + p_\text{m}) = -2 \lambda (d-1) (\ddot{H}+d H\dot{H})\, .
\end{eqnarray}
We take the form of
\begin{eqnarray}
    \rho_\text{m} (t) = \frac{2 (d-1) \ell}{\kappa \lambda_V^2 \left[ 1 - \frac{2\gamma(d-1)}{(d-2)}(1 + w_\text{m})\right] t^2} \left[\frac{1}{1 - \frac{2\gamma d}{d-2}} - 1\right]\, ,
\end{eqnarray}
for non-vacuum fluid $(w_\text{m} \neq -1)$ and
\begin{eqnarray}
    \rho_\text{m} (t) = \frac{4 \gamma (d-1) (2 - d\, \ell)}{(d - 2)\kappa \lambda_V^2 t^2} \, ,
\end{eqnarray}
for vacuum fluid $(w_\text{m} = -1)$ in our calculation and use the form of $H(t) = \ell/t$ to obtain an exact solution of power-law inflation. The form of $\phi(t)$, $H(t)$ and  $\rho_\text{m} (t)$ are consistent to solve Eq. \eqref{eq:friedmanninf} and Eq. \eqref{eq:C} with the consequence that
\begin{eqnarray} \label{eq:ell}
    \ell = \frac{2 (d-2)}{(d - 2 - 2\gamma d)^2 \lambda_V^2} \Bigg[ 1 + \sqrt{ 1 - \frac{4\gamma (d - 2 - 2\gamma d)^3}{(d - 2)^3} \lambda_V^2} \Bigg] > 1 \, ,
\end{eqnarray}
for $w_\text{m} \neq -1$ case and 
\begin{eqnarray} \label{eq:ellvac}
    \ell = \frac{4}{(d - 2) \lambda_V^2} > 1 \, ,
\end{eqnarray}
 for $w_\text{m} = -1$ case. It can be seen that the solution in Eq. \eqref{eq:ell} only exists for $0 < \lambda_V \leq \sqrt{\frac{(d-2)^3}{4\gamma (d-2-2\gamma d)^3}}$. It is possible to obtain $\ell \gg 1$ with the corresponding parameters $\lambda_V$ and $\gamma$  satisfying the slow-roll condition since $\epsilon = \frac{|\dot{H}|}{H^2} = \frac{1}{\ell} \ll 1$. For instance, it can be seen in Fig. \ref{fig:pli}. It is worth mentioning that power-law inflation discussed in the GR framework has been ruled out by PLANCK 2013 data \cite{ade2014planck} such that, in this sense, it is viewed as a failed model of inflation. Fortunately, there is an attempt to solve this problem in terms of cuscuton-gravity \cite{ito2019dressed} in which it can satisfy PLANCK constraint such that it may remain as an accepted inflation model. We roughly suspect that power-law inflation may be recovered in theories beyond GR, one of which is in the KGR cosmology, although further investigation is needed to ensure this statement. \\ \\
\begin{figure}[h]
	\centering
	\includegraphics[width=0.5\textwidth]{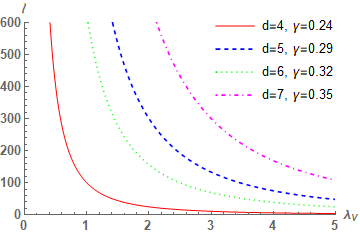}
	\caption{The values of $\ell$ as a function of $\lambda_V$.}
	\label{fig:pli}    
\end{figure}
\subsection{The Late-Time Universe: Dark Energy Paradigm in KGR Cosmology}
At the same fixed point, CP4, it is possible to obtain stable nodes and an accelerating universe when the scalar field and non-relativistic matter $(w_\text{m} = 0)$ exist in KGR cosmology for $\Lambda = 0$ case. In this sense, the accelerating universe can be explained by a nontrivial interplay between a scalar field with an exponential potential and the non-conservative EMT of the non-relativistic matter field (baryonic and dark matter) in curved spacetime which plays a role as dark energy since $w_\text{KGR} < -\frac{1}{(d-1)}$. This result agrees with the cosmological observation that the subdominant components of our current universe are baryonic and dark matter, $\Omega_\text{m} \approx 0.3$ (See Fig. \ref{fig:cp4t}). In addition, the phase plane illustrations for the $\Lambda = 0$ case are provided in Fig. \ref{fig:phase}. We may argue that our model can alleviate the cosmic coincidence problem since our observed universe is a late-time attractor. However, these features cannot be found for the $\Lambda\neq 0$ case due to the absence of stable nodes in CP4. The late-time behavior for this case is given by the interplay between massless scalar field and cosmological constant acting as dark energy taking responsibility for the accelerated expansion of the universe.\\
\subsection{Cosmological Sequence of KGR Theory}
The KGR theory without a cosmological constant might provide a cosmological sequence which starts from inflation and leads to the dark energy era. Utilizing information that the subdominant components of the current universe are represented by cold dark matter and baryons, then $w_\text{m} = 0$ \cite{kopp2018dark}. Thus, in the late-time sector we can set $w_\text{m} = 0$ allowing us to obtain a scalar field-dust-like matter (baryonic and dark matter) universe as a late-time attractor. It may eliminate the cosmic coincidence problem. However, in the early-time inflation sector, observational data regarding $w_\text{m}$ are not available. Therefore, in this discussion, we choose the value of $w_\text{m}$ based on inflation criteria from a dynamical system perspective, which involves unstable critical points and accelerating universe. From Fig. \ref{fig:cp4t}, both unstable and stable critical points that allow the universe to expand at an accelerated rate are obtained from CP4 with particular values of $\gamma$ depending on dimension $d$.

In this discussion, without loss of generality, we will only mention the possible cosmic history for specific values of cosmological parameters chosen in dimensions $d=4$, $d=5$, and $d=6$, along with their implications. In the case of selecting a parameter set $(d, \gamma, \lambda_V) = (4, 0.23, 2.3)$, we may have a sequence of cosmological eras that includes both power-law inflation and late-time acceleration with $w_\text{m} = 1$ and  $w_\text{m} = 0$, respectively. Similar features can also be obtained for higher dimensions, for instance, in the case of $(d, \gamma, \lambda_V) = (5, 0.29, 3.4)$, inflation and dark energy eras are achieved if $w_\text{m} = 0.45$ and $w_\text{m} = 0$, respectively. As a final example, for $(d, \gamma, \lambda_V) = (6, 0.32, 3.6)$, $w_\text{m} = 0.28$ and $w_\text{m} = 0$ correspond to saddle power-law inflation and late-time acceleration, respectively. Refer to Fig. \ref{fig:phase} for phase plane diagrams. Note that we can vary the value of the parameter $w_\text{m}$ in the early-time era while still permitting inflation and dark energy eras to occur. 

Furthermore, all the speculative scenarios mentioned above allow the universe to pass through radiation-dominated and matter-dominated eras. The radiation-dominated era, with $w_\text{m} = \frac{1}{(d-1)}$ and $w_\text{eff} = \frac{1}{(d-1)}$, can be obtained without any issues in this model. However, for the matter-dominated era, some attention is needed, as for $w_\text{m} = 0$, the universe would experience accelerated expansion due to the coupling between non-relativistic matter and geometry, effectively causing cold dark matter to behave like dark energy. Nevertheless, to ensure that the universe has passed through a matter-dominated era, we need to ensure that the solution at CP1 allows us to obtain $w_\text{eff}=0$ (see Fig. \ref{fig:cp1}). This condition can be achieved for $0 < w_\text{m} < \frac{1}{(d-1)}$, which means that the universe is dominated by warm dark matter \cite{wei2013cosmological}, which effectively behaves like cold dark matter in the Rastall framework. These results are precisely as predicted in \cite{chagoya2023cosmological}, suggesting that cold dark matter in the Rastall framework might behave akin to warm/hot dark matter in the GR framework. It is worth stressing that the cosmological sequence in the KGR theory may be summarized as: inflation (CP4-unstable-accelerated) $\rightarrow$ radiation-dominated (CP1-saddle-decelerated) $\rightarrow$ matter-dominated (CP1-saddle-decelerated) $\rightarrow$ dark energy (CP4-stable-accelerated). To achieve a more realistic cosmic history, which involves merging inflation and dark energy scenarios through radiation and matter-dominated eras which are consistent with observational outcomes, this model still requires verification from observational data on the parameters $w_\text{m}$ and $\gamma$, especially during the inflation and matter-dominated epochs.

Lastly, we can also track the cosmic history of KGR cosmology for the case of $\Lambda \neq 0$, which is compatible only if $\lambda_V = 0$. This is equivalent to the context of a massless scalar field and a cosmological constant. The radiation and matter-dominated eras are part of cosmological sequence of this model, with a late-time attractor being an accelerating universe driven by dark energy. Thus, the chronology of the universe for this case may be summarized as: radiation-dominated (CP1-saddle-decelerated) $\rightarrow$ matter-dominated (CP1-saddle-decelerated) $\rightarrow$ dark energy-dominated (CP5-stable-accelerated).
\begin{figure}[!ht]
	\centering
		\begin{subfigure}[t]{0.49\textwidth}
		\includegraphics[width=1\textwidth]{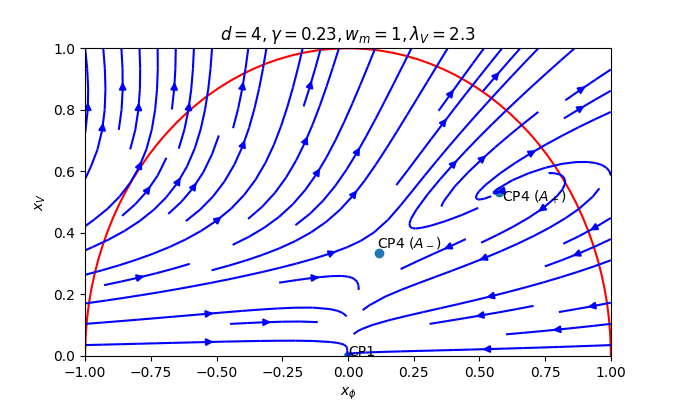}
		\caption{$d=4$, $\gamma = 0.23$, $w_\text{m} = 1$, $\lambda_V=2.3$}
	\end{subfigure}   
	\begin{subfigure}[t]{0.49\textwidth}
		\includegraphics[width=1\textwidth]{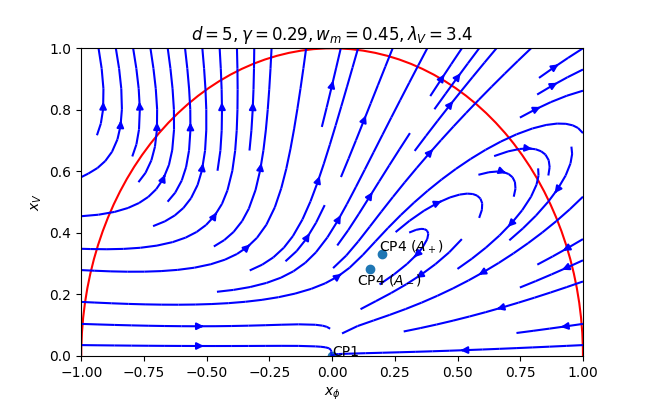}
		\caption{$d=5$, $\gamma = 0.29$, $w_\text{m} = 0.45$, $\lambda_V=3.4$}
	\end{subfigure} 
	\begin{subfigure}[t]{0.49\textwidth}
		\includegraphics[width=1\textwidth]{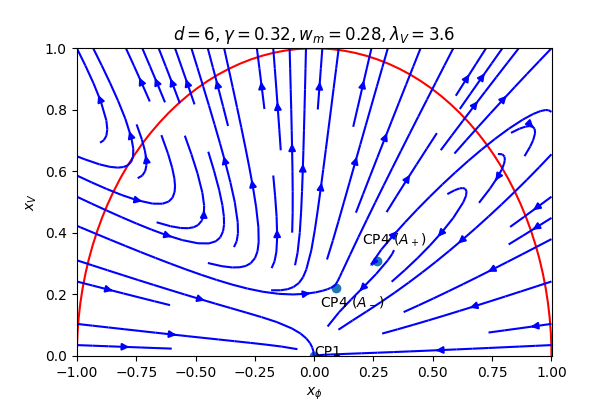}
		\caption{$d=6$, $\gamma = 0.32$, $w_\text{m} = 0.28$, $\lambda_V=4.6$}
	\end{subfigure}
	\begin{subfigure}[t]{0.49\textwidth}
		\includegraphics[width=1\textwidth]{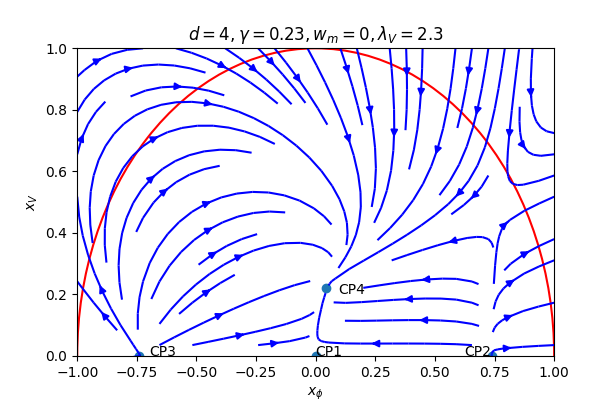}
		\caption{$d=4$, $\gamma = 0.23$, $w_\text{m} = 0$, $\lambda_V=2.3$}
	\end{subfigure}   
	\begin{subfigure}[t]{0.49\textwidth}
		\includegraphics[width=1\textwidth]{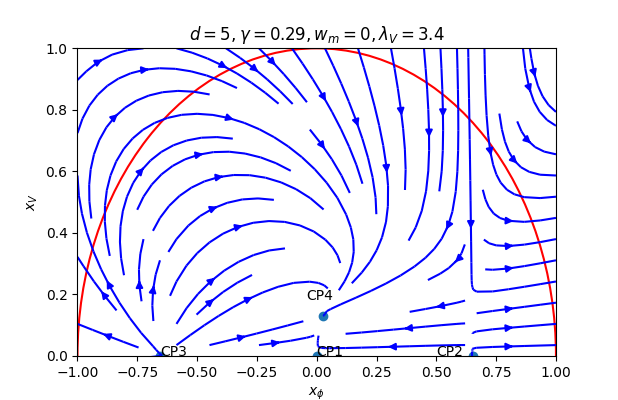}
		\caption{$d=5$, $\gamma = 0.29$, $w_\text{m} = 0$, $\lambda_V=3.4$}
	\end{subfigure} 
	\begin{subfigure}[t]{0.49\textwidth}
		\includegraphics[width=1\textwidth]{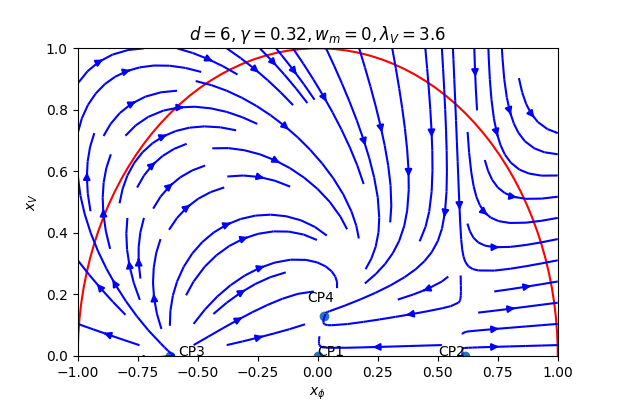}
		\caption{$d=6$, $\gamma = 0.32$, $w_\text{m} = 0$, $\lambda_V=3.6$}
	\end{subfigure}
	\caption{The phase plane of the cosmological model KGR for early-time (a) $d=4, \gamma=0.23, w_\text{m}=1, \lambda_V=2.3$, (b) $d=5, \gamma=0.29, w_\text{m}=0.45, \lambda_V=3.4$, and (c) $d=6, \gamma=0.32, w_\text{m}=0.28, \lambda_V=3.6$, and late-time (d) $d=4, \gamma=0.23, w_\text{m}=0, \lambda_V=2.3$, (e) $d=5, \gamma=0.29, w_\text{m}=0, \lambda_V=3.4$, (f) $d=6, \gamma=0.32, w_\text{m}=0, \lambda_V=3.6$ for the case of $\Lambda = 0$. We have identified unphysical stable critical points (CP1) in Subfig. (a), (b), and (c) due to the condition $\Omega_\text{m} < 0$ (see Fig.\ref{fig:cp1}). The universe with the composition of the scalar field and dust-like matter (baryonic and dark matter) constitutes a late-time attractor. The red line in the pictures represent a semicircle with a radius of one.} 
	\label{fig:phase}    
\end{figure}  
\section{Conclusion}
\label{sec:conclusion}
We have investigated the cosmological consequences of spatially flat FLRW spacetime in higher dimensional KGR theory using dynamical system analysis. A scalar field with exponential potential is chosen to transform Friedmann equations into autonomous equations. We find that there exist five critical points, namely CP$_{1-5}$, which can be related to the cosmological era based on its stability, equation of state, density parameter, and decelerated parameter.

Then, we have established the local-global existence and
the uniqueness of the evolution equations \eqref{eq:xphi} and \eqref{eq:xv} with constraint \eqref{constraineq} for $\Lambda = 0$ case and of the equations \eqref{eq:xphia}, \eqref{eq:xva} and \eqref{eq:xLambdaa} with constraint \eqref{constraineq1} for $\Lambda \neq 0$ using Picard’s iteration and the contraction mapping properties. In each case, we consider all ranges of parameters $w_\text{m}$ and $\gamma$ correspond to the signature of each coefficient in Eq.\eqref{constraineq} and  \eqref{constraineq1}. Note that our results apply to both cases, namely $\Lambda = 0$ and $\Lambda \neq 0$.

We have also particularly discussed several possible cosmological models of higher dimensional KGR theory for the $\Lambda = 0$ case that are suitable to explain both inflation and the accelerated universe in the late-time. The foundation of our investigation of inflation within this model relies on the existence of unstable solutions from CP4.  The reason for selecting the term "unstable critical point" to characterize this area lies in its potential to reveal an intrinsic instability mechanism within the KGR theory. This mechanism may serve as a graceful way to exit the inflation process (See Ref. \cite{odintsov2015singular,nojiri2015singular,oikonomou2018autonomous} for further discussion). Another interesting result of KGR inflation is that when compared to the original PLI model \cite{halliwell1987scalar}, apart from its instability, we can also obtain inflationary solutions even with steep potential. The reason for exploring such types of potentials is based on the motivation that certain physics theories make such predictions. Previously, there have been several studies seeking inflationary solutions for steep potentials, such as multifield inflation \cite{liddle1998assisted} or by introducing an interaction coupling between the inflaton and radiation \cite{wands1993exponential,yokoyama1988dynamics, billyard2000interactions}. 

In the inflationary sector, we have derived the exact solution of power-law inflation without assuming a slow-roll mechanism in which it represents a scalar field-perfect fluid universe given by CP4. It is necessary to investigate further whether the power-law inflation in the KGR model can satisfy the PLANCK constraints. In the late-time sector, CP4 of the KGR model also offers an interesting feature such that cosmic acceleration occurs due to non-trivial interplay between a scalar field with an exponential potential and non-conservation of EMT of baryonic and dark matter which play a role as dark energy since $w_\text{KGR} < -\frac{1}{(d-1)}$. These results indicate that the KGR model with a scalar field having an exponential potential, $\lambda_V \neq 0$, without the presence of a cosmological constant, is a viable cosmological model in both the early-time and late-time eras. Under these circumstances, the cosmological sequence within the KGR theory may be outlined as: inflation results in a period dominated by radiation, which subsequently shifts to an era dominated by matter, eventually giving rise to the presence of dark energy.

On the other hand, the KGR model for $\Lambda \neq 0$ remains viable to produce power-law inflation during the early-time era, whereas, in the late-time era, it is only compatible if $\lambda_V = 0$, which represents the universe dominated by the constant potential term of the scalar field and the cosmological constant $\Lambda$. This implies that the KGR model involving a scalar field with an exponential potential, $\lambda_V \neq 0$, and the existence of a cosmological constant is incompatible with a physical cosmological model in the late-time era. The timeline of the universe in this case may be sketched as: radiation-dominated leading to a matter-dominated phase, ultimately transitioning into a dark energy-dominated universe.
\section*{Acknowledgments}
TAW acknowledges LPDP for financial support. The work of BEG and AS is supported by Hibah Riset ITB. BEG also acknowledges  Hibah Riset Fundamental Kemendikbudristekdikti for financial support. HA is partially supported by the World Class Research (WCR) Grant from Kemendikbudristek-IPB for fiscal year 2022.

\appendix

\section{Appendix A: Linear Stability for $\Lambda = 0$ Case}
\label{sec:appA}
 In this Appendix, we analyze the linear perturbation of the equations \eqref{eq:xphi}-\eqref{eq:xv} around the critical points $(x_{\phi,c}, x_{V,c})$ by expanding the autonomous variables as follow
 \begin{eqnarray}
     x_{\phi}=x_{\phi,c} + u_{\phi}\, ,
 \end{eqnarray}
 \begin{eqnarray}
     x_{V}=x_{V,c} + u_{V}\, .
 \end{eqnarray}
The form of the equation of motions for the first order is as follows
\begin{align}
    \frac{2}{(d-1)N} u_{\phi}^{'} = &\left[\frac{(1 + w_\text{m})\left(1-\frac{2\gamma d}{(d-2)} - x_{V,c}^2 \right) + 3(1-w_\text{m}) x_{\phi,c}^2}{\left[1-\frac{2\gamma(d-1)}{(d-2)}(1+w_\text{m})\right]}-2 \right] u_{\phi}\, \nonumber \\
    &-\left[\frac{2(1+w_m)}{\left[1-\frac{2\gamma(d-1)}{(d-2)}(1+w_m)\right]}x_{\phi,c} - 2\lambda_V \sqrt{\frac{d-2}{d-1}} \, \right]x_{V,c}\, u_V \, ,
\end{align}
\begin{align}
    \frac{2}{(d-1)N} u_{V}^{'} = & \Bigg[\frac{(1 + w_\text{m}) \left(1-\frac{2\gamma d}{(d-2)} - 3 x_{V,c}^2 \right) + (1-w_\text{m}) x_{V,c}^2}{\left[1-\frac{2\gamma(d-1)}{(d-2)}(1+w_\text{m})\right]} - \lambda_V \sqrt{\frac{d-2}{d-1}}  x_{\phi,c}\, \Bigg]u_{V}\, \nonumber \\
    &+ \left[\frac{2(1 - w_\text{m})}{\left[1-\frac{2\gamma(d-1)}{(d-2)}(1+w_\text{m})\right]}x_{\phi,c} - \lambda_V \sqrt{\frac{d-2}{d-1}}\, \right]x_{V,c}\, u_{\phi}\, .
\end{align}
Therefore, we can express the equations mentioned above in a matrix form
\begin{gather}
\begin{pmatrix}
    u_{\phi}^{'}\\
    u_{V}^{'}\\
\end{pmatrix}
=\textbf{\textit{J}}
\begin{pmatrix}
    u_{\phi}\\
    u_{V}\\ 
\end{pmatrix} \, , 
\end{gather}
 where $\textit{\textbf{J}}$ is Jacobian matrix. Its eigenvalues, $\mu_1$ and $\mu_2$, will be identified and recorded for each critical point. These values will be used to evaluate the stability characteristics of the critical points. The Jacobian matrix along with its eigenvalues for each critical point of the autonomous system in this case can thus be written as follow
 \begin{itemize}
    \item Jacobian matrix for CP1 $(0,0)$:
    \begin{gather}
    \textbf{\textit{J}}=
    \begin{pmatrix}
        \frac{\left[1-\frac{2\gamma d}{(d-2)}\right](1+w_\text{m})}{\left[1-\frac{2\gamma(d-1)}{(d-2)}(1+w_\text{m})\right]}-2& 0 \\ \\
        0 & \frac{\left[1-\frac{2\gamma d}{(d-2)}\right](1+w_\text{m})}{\left[1-\frac{2\gamma(d-1)}{(d-2)}(1+w_\text{m})\right]}
    \end{pmatrix} \, ,
    \end{gather}
with eigenvalues 
 \begin{equation}
        \mu_1=\frac{\left[1-\frac{2\gamma d}{(d-2)}\right](1+w_\text{m})}{\left[1-\frac{2\gamma(d-1)}{(d-2)}(1+w_\text{m})\right]}-2\, , \quad    \mu_2 = \frac{\left[1-\frac{2\gamma d}{(d-2)}\right](1+w_\text{m})}{\left[1-\frac{2\gamma(d-1)}{(d-2)}(1+w_\text{m})\right]}\, .
    \end{equation}
     \item Jacobian matrix for CP2 $\Bigg(\sqrt{1-\frac{2\gamma(1+w_\text{m})}{(1 - w_\text{m})}} , 0\Bigg)$ :
    \begin{gather}
    \textbf{\textit{J}}=
    \begin{pmatrix}
        \frac{2[(1-w_\text{m})-2\gamma(1+w_\text{m})]}{\left[1-\frac{2\gamma(d-1)}{(d-2)}(1+w_\text{m})\right]}& 0 \\\\
        0 & 2-\lambda_V\sqrt{\frac{d-2}{d-1}}\sqrt{1-\frac{2\gamma(1+w_\text{m})}{(1-w_\text{m})}} \\ \\
    \end{pmatrix}\, ,
    \end{gather}
with eigenvalues   
\begin{equation}
       \mu_1=\frac{2[(1-w_\text{m})-2\gamma(1 + w_\text{m})]}{\left[1-\frac{2\gamma(d-1)}{(d-2)}(1+w_\text{m})\right]} \, , \quad  \mu_2 = 2 - \lambda_V\sqrt{\frac{d-2}{d-1}}\sqrt{1-\frac{2\gamma(1+w_\text{m})}{(1-w_\text{m})}}\, .
    \end{equation}
         \item Jacobian matrix for CP3 $\Bigg(-\sqrt{1-\frac{2\gamma(1+w_\text{m})}{(1 - w_\text{m})}} , 0\Bigg)$ :
    \begin{gather}
    \textbf{\textit{J}}=
    \begin{pmatrix}
        \frac{2[(1-w_\text{m})-2\gamma(1+w_\text{m})]}{\left[1-\frac{2\gamma(d-1)}{(d-2)}(1+w_m)\right]}& 0 \\\\
        0 & 2 + \lambda_V\sqrt{\frac{d-2}{d-1}}\sqrt{1-\frac{2\gamma(1+ w_\text{m})}{(1- w_\text{m})}} \\ \\
    \end{pmatrix}\, ,
    \end{gather}
with eigenvalues   
\begin{equation}
       \mu_1=\frac{2[(1-w_\text{m})-2\gamma(1 + w_\text{m})]}{\left[1-\frac{2\gamma(d-1)}{(d-2)}(1+w_\text{m})\right]} \, , \quad  \mu_2 = 2 + \lambda_V\sqrt{\frac{d-2}{d-1}}\sqrt{1-\frac{2\gamma(1 + w_\text{m})}{(1 - w_\text{m})}}\, .
    \end{equation}
     \item Jacobian matrix for CP4 $\Bigg(\frac{1}{\lambda_V} \sqrt{\frac{d-1}{d-2}} \mathcal{A}_{\pm} , \frac{1}{\lambda_V} \sqrt{\frac{(d-1)(2 - \mathcal{A}_{\pm}) 
     \mathcal{A}_{\pm}}{(d-2)}}\Bigg)$:
\begin{gather}
	\hspace{-0.5cm}\textbf{\textit{J}}=
	\begin{pmatrix}
	\mathcal{A_{\pm}}+\frac{2(d-1)(1-w_\text{m})A_{\pm}^2}{\lambda_V^2(d-2)\left[1-\frac{2\gamma(d-1)}{(d-2)}(1+w_\text{m})\right]}-2 \quad & f(\mathcal{A_{\pm}})\left[2-\frac{2(d-1)(1+w_\text{m})\mathcal{A}_{\pm}}{\lambda_V^2(d-2)\left[1-\frac{2\gamma(d-1)}{(d-2)}(1+w_\text{m})\right]}\right]\\ \\
	f(\mathcal{A_{\pm}})\left[-1+\frac{2(d-1)(1-w_\text{m})\mathcal{A}_{\pm}}{\lambda_V^2(d-2)\left[1-\frac{2\gamma(d-1)}{(d-2)}(1+w_\text{m})\right]}\right] \quad & -\frac{2(d-1)(1+w_\text{m})(f(\mathcal{A_{\pm}}))^2}{\lambda_V^2(d-2)\left[1-\frac{2\gamma(d-1)}{(d-2)}(1+w_\text{m})\right]}
	\end{pmatrix}\, ,
	\end{gather}
    where $f(\mathcal{A}_{\pm}) = \sqrt{(2-\mathcal{A}_{\pm})\mathcal{A}_{\pm}}$.
    \item Jacobian matrix for CP5 $\Bigg(0,\sqrt{1-\frac{2\gamma d}{(d-2)}}\Bigg)$: 
        \begin{gather}
    \textbf{\textit{J}}=
    \begin{pmatrix}
        -2 \quad &0 \\ \\
        0 & -\frac{2(1+w_\text{m})}{\left[1-\frac{2\gamma(d-1)}{(d-2)}(1+w_\text{m})\right]}\left[1-\frac{2\gamma d}{(d-2)}\right]&
    \end{pmatrix}\, ,
    \end{gather}
    with eigenvalues
     \begin{equation}
        \mu_1=-2 \, ,  \quad \mu_2=-\frac{2(1+w_\text{m})\left[1-\frac{2\gamma d}{(d-2)}\right]}{\left[1-\frac{2\gamma(d-1)}{(d-2)}(1+w_\text{m})\right]}\, .
    \end{equation}
  \end{itemize}  
  \section{Appendix B: Linear Stability for $\Lambda \neq 0$ Case}
\label{sec:appB}
 In this Appendix, we analyze the linear perturbation of the equations \eqref{eq:xphia}, \eqref{eq:xva} and \eqref{eq:xLambdaa} around the critical points $(x_{\phi,c}, x_{V,c},x_{\Lambda,c\,(\pm)})$ by expanding the autonomous variables as follow
 \begin{eqnarray}
     x_{\phi}=x_{\phi,c} + u_{\phi}\, ,
 \end{eqnarray}
 \begin{eqnarray}
     x_{V}=x_{V,c} + u_{V}\, ,
 \end{eqnarray}
 \begin{eqnarray}
     x_{\Lambda\, (\pm)} = x_{\Lambda,c\,(\pm)} + u_{\Lambda\, (\pm)}\, .
 \end{eqnarray}
The form of the equation of motions for the first order is as follows
\begin{align}
    \frac{2}{(d-1)N} u_{\phi}^{'} = &\left[\frac{(1 + w_\text{m})\left(1-\frac{2\gamma d}{(d-2)} \mp x_{V\Lambda,c\, (\pm)}^2 \right) + 3(1-w_\text{m}) x_{\phi,c}^2}{\left[1-\frac{2\gamma(d-1)}{(d-2)}(1+w_\text{m})\right]}-2 \right] u_{\phi}\, \nonumber \\
    &-\left[\frac{2(1+w_m)}{\left[1-\frac{2\gamma(d-1)}{d-2}(1+w_m)\right]}x_{\phi,c} - 2 \lambda_V \sqrt{\frac{d-2}{d-1}} \, \right]x_{V,c}\, u_V \, \nonumber \\
    &-\frac{2(1+w_\text{m})}{\left[1-\frac{2\gamma(d-1)}{(d-2)}(1+w_m)\right]}x_{V,c}x_{\Lambda,c\, (\pm)}u_{\Lambda\,(\pm)}\, , \\ \nonumber \\
    \frac{2}{(d-1)N} u_{V}^{'} = & \Bigg[\frac{(1 + w_\text{m}) \left(1-\frac{2\gamma d}{d-2} - 3 x_{V,c}^2 \mp x_{\Lambda,c\, (\pm)}^2 \right) + (1-w_\text{m}) x_{\phi,c}^2}{\left[1-\frac{2\gamma(d-1)}{(d-2)}(1+w_m)\right]} - \lambda_V \sqrt{\frac{d-2}{d-1}} x_{\phi,c}\, \Bigg]u_{V}\, \nonumber \\
    &+ \left[\frac{2(1 - w_\text{m})}{\left[1-\frac{2\gamma(d-1)}{(d-2)}(1+w_\text{m})\right]}x_{\phi,c} - \lambda_V \sqrt{\frac{d-2}{d-1}}\, \right]x_{V,c}\, u_{\phi}\, \nonumber \\
    &-\frac{2(1 + w_\text{m})}{\left[1-\frac{2\gamma(d-1)}{(d-2)}(1+w_\text{m})\right]}x_{V,c}\, x_{\Lambda,c\, (\pm)}\, u_{\Lambda\, (\pm)}\, , \\ \nonumber \\ 
    \frac{2}{(d-1)N} u_{\Lambda\, (\pm)}^{'} = &\frac{1}{\left[1-\frac{2\gamma(d-1)}{(d-2)}(1 + w_\text{m})\right]}\Bigg[2 \Big[(1 - w_\text{m})x_{\phi,c}\, \, u_{\phi} - 2(1+w_\text{m})\, x_{V,c}\, u_{V} \Big] x_{\Lambda,c \, (\pm)}\, \nonumber  \\
    &+\Bigg[(1 + w_\text{m}) \left(1-\frac{2\gamma d}{(d-2)} -x_{V,c}^2\right)+(1-w_\text{m})x_{\phi,c}^2\, \nonumber \\ 
&\mp 3 (1 + w_\text{m}) \, x_{\Lambda,c \, (\pm)}^2\Bigg]u_{\Lambda}\Bigg]\, .
\end{align}
Therefore, we can express the equations mentioned above in a matrix form
\begin{gather}
\begin{pmatrix}
    u_{\phi}^{'}\\
    u_{V}^{'}\\
    u_{\Lambda\, (\pm)}^{'}
\end{pmatrix}
=\textbf{\textit{J}}
\begin{pmatrix}
    u_{\phi}\\
    u_{V}\\ 
    u_{\Lambda\, (\pm)}
\end{pmatrix} \, , 
\end{gather}
 where $\textit{\textbf{J}}$ is the Jacobian matrix. Its eigenvalues, $\mu_1$, $\mu_2$, and $\mu_3$, will be identified and recorded for each critical point. These values will be used to evaluate the stability characteristics of the critical points. The Jacobian matrix along with its eigenvalues for each critical point of the autonomous system in this case can thus be written as follow
 \begin{itemize}
    \item Jacobian matrix for CP1 $(0,0,0)$:
       \begin{gather}
    \textbf{\textit{J}}=
    \begin{pmatrix}
        \frac{\left[1-\frac{2\gamma d}{(d-2)}\right](1 + w_\text{m})}{\left[1-\frac{2\gamma(d-1)}{(d-2)}(1 + w_\text{m})\right]}-2& 0 &0\\ \\
        0 & \frac{\left[1-\frac{2\gamma d}{(d-2)}\right](1 + w_\text{m})}{\left[1-\frac{2\gamma(d-1)}{(d-2)}(1 + w_\text{m})\right]}&0 \\ \\
        0 & 0 & \frac{\left[1-\frac{2\gamma d}{(d-2)}\right](1 + w_\text{m})}{\left[1-\frac{2\gamma(d-1)}{(d-2)}(1 + w_\text{m})\right]}
    \end{pmatrix}\, ,
    \end{gather}
with eigenvalues 
 \begin{equation}
        \mu_1=\frac{\left[1-\frac{2\gamma d}{(d-2)}\right](1+w_\text{m})}{\left[1-\frac{2\gamma(d-1)}{(d-2)}(1+w_\text{m})\right]}-2\, , \quad    \mu_2 = \mu_3=\frac{\left[1-\frac{2\gamma d}{(d-2)}\right](1+w_\text{m})}{\left[1-\frac{2\gamma(d-1)}{(d-2)}(1+w_\text{m})\right]}\, .
    \end{equation}
     \item Jacobian matrix for CP2 $\Bigg(\sqrt{1-\frac{2\gamma(1+w_\text{m})}{(1 - w_\text{m})}} , 0, 0\Bigg)$ :
   \begin{gather}
    \textbf{\textit{J}}=
    \begin{pmatrix}
        \frac{2[(1-w_m)-2\gamma(1+w_m)]}{\left[1-\frac{2\gamma(d-1)}{(d-2)}(1+w_m)\right]}& 0 &0\\\\
        0 & 2-\lambda_V\sqrt{\frac{d-2}{d-1}}\sqrt{1-\frac{2\gamma(1+w_m)}{(1-w_m)}} & 0 \\ \\
        0 &0 &2
    \end{pmatrix}\, ,
    \end{gather}
with eigenvalues   
\begin{align}
       \mu_1 = & \, \frac{2[(1-w_\text{m})-2\gamma(1 + w_\text{m})]}{\left[1-\frac{2\gamma(d-1)}{(d-2)}(1+w_\text{m})\right]} \, , \nonumber \\ 
       \mu_2 = & \, 2 - \lambda_V\sqrt{\frac{d-2}{d-1}}\sqrt{1-\frac{2\gamma(1+w_\text{m})}{(1-w_\text{m})}}\, , \\
       \mu_3 =  & \, 2\, . \nonumber
\end{align}
         \item Jacobian matrix for CP3 $\Bigg(-\sqrt{1-\frac{2\gamma(1+w_\text{m})}{(1 - w_\text{m})}} , 0, 0\Bigg)$ :
\begin{gather}
    \textbf{\textit{J}} = 
    \begin{pmatrix}
        \frac{2[(1-w_m)-2\gamma(1+w_m)]}{\left[1-\frac{2\gamma(d-1)}{(d-2)}(1+w_m)\right]}& 0 &0 \\ \\
        0 & 2 + \lambda_V \sqrt{\frac{d-2}{d-1}} \sqrt{1-\frac{2\gamma(1+w_m)}{(1-w_m)}} & 0\\ \\
        0 &0 &2
    \end{pmatrix}\, ,
\end{gather}
with eigenvalues   
\begin{align}
       \mu_1 = & \, \frac{2[(1-w_\text{m})-2\gamma(1 + w_\text{m})]}{\left[1-\frac{2\gamma(d-1)}{(d-2)}(1+w_\text{m})\right]} \, , \nonumber \\ 
       \mu_2 = & \, 2 + \lambda_V\sqrt{\frac{d-2}{d-1}}\sqrt{1-\frac{2\gamma(1+w_\text{m})}{(1-w_\text{m})}}\, , \\
       \mu_3 =  & \, 2\, . \nonumber
\end{align}
     \item Jacobian matrix for CP4 $\Bigg(\frac{1}{\lambda_V} \sqrt{\frac{d-1}{d-2}} \mathcal{A}_{\pm} , \frac{1}{\lambda_V} \sqrt{\frac{(d-1)(2 - \mathcal{A}_{\pm}) 
     \mathcal{A}_{\pm}}{(d-2)}}, 0\Bigg)$:
\begin{gather}
	\hspace{-0.4cm}\textbf{\textit{J}}=
	\begin{pmatrix}
	\mathcal{A_{\pm}}+\frac{2(d-1)(1-w_\text{m})A_{\pm}^2}{\lambda_V^2(d-2)\left[1-\frac{2\gamma(d-1)}{(d-2)}(1+w_\text{m})\right]}-2 \quad & f(\mathcal{A}_{\pm})\left[2-\frac{2(d-1)(1+w_\text{m})\mathcal{A}_{\pm}}{\lambda_V^2(d-2)\left[1-\frac{2\gamma(d-1)}{(d-2)}(1+w_\text{m})\right]}\right] \quad & 0\\ \\
	f(\mathcal{A}_{\pm})\left[-1+\frac{2(d-1)(1-w_\text{m})\mathcal{A}_{\pm}}{\lambda_V^2(d-2)\left[1-\frac{2\gamma(d-1)}{(d-2)}(1+w_\text{m})\right]}\right] \quad & -\frac{2(d-1)(1+w_\text{m})(2-\mathcal{A}_{\pm})\mathcal{A}_{\pm}}{\lambda_V^2(d-2)\left[1-\frac{2\gamma(d-1)}{(d-2)}(1+w_\text{m})\right]} \quad & 0 \\ \\
	0 \quad & 0 \quad & \mathcal{A}_{\pm}
	\end{pmatrix}\, ,
\end{gather}
where $f(\mathcal{A}_{\pm}) = \sqrt{(2-\mathcal{A}_{\pm})\mathcal{A}_{\pm}}$.
    \item Jacobian matrix for CP5 $\left(0,\sqrt{1-\frac{2\gamma d}{(d-2)}\mp x_{\Lambda,c(\pm)}^2},x_{\Lambda,c \, (\pm)}\right)$: 
   \begin{gather}
    \textbf{\textit{J}}=
    \begin{pmatrix}
        -2 \hspace{0.8cm} &0 & 0\\\\
        0 & -g\left[1-\frac{2\gamma d}{d-2} \mp x_{\Lambda,c\, (\pm)}^2\right]& -g\sqrt{1-\frac{2\gamma d}{d-2} \mp x_{\Lambda,c\, (\pm)}^2}\, x_{\Lambda,c\, (\pm)}\\\\
        0 & -g \sqrt{1-\frac{2\gamma d}{d-2}-x_{\Lambda,c\, (\pm)}^2}\, x_{\Lambda,c\, (\pm)} & \mp g\, x_{\Lambda,c\, (\pm)}^2
    \end{pmatrix}
    \end{gather}
    with eigenvalues
    \begin{equation}
        \mu_1=-2 \, ,  \quad \mu_2=-g\, {\left[1-\frac{2\gamma d}{(d-2)}\right]}\, , \quad \mu_3 =0\, ,
    \end{equation}
    where $g = \frac{2(1+w_m)}{\left[1-\frac{2\gamma(d-1)}{(d-2)}(1+w_m)\right]}$.
  \end{itemize}

\end{document}